\documentclass[11pt]{article}
\usepackage{geometry}                
\usepackage{graphicx}
\usepackage{amsmath,amssymb, amsthm,epsfig,bm}
\usepackage{setspace}
\usepackage{epstopdf}
\usepackage{psfrag}
\usepackage{color,soul}
\usepackage{verbatim}
\usepackage{algorithm}
\usepackage{algpseudocode}
\usepackage{upgreek}
\usepackage{multirow}
\usepackage{natbib}
\usepackage{supertabular}
\usepackage{xr}
\usepackage{hyperref}
\hypersetup{
    colorlinks=true,
    linkcolor=blue,
    filecolor=magenta,      
    urlcolor=cyan,
}
\bibpunct{(}{)}{;}{a}{}{,} 
\sethlcolor{yellow}
\theoremstyle{plain} 
\newtheorem{theorem}{Theorem}

\newtheorem{proposition}{Proposition}

\theoremstyle{definition}
\newtheorem{definition}{Definition}
\newtheorem{remark}{Remark}

\newcommand{\bfH}{\mathbf{H}}

\newcommand{\bmphi}{\boldsymbol{\phi}}

\newcommand{\bmbeta}{\boldsymbol{\beta}}

\newcommand{\bfzr}{\mathbf{0}}
\newcommand{\calG}{\mathcal{G}}

\newcommand{\supp}{\text{supp}}
\newcommand{\defi}{\mathop{=}\limits^{\Delta}}      

\providecommand{\abs}[1]{\lvert#1\rvert}
\providecommand{\norm}[1]{\lVert#1\rVert}

\begin{document}

\title{Penalized Estimation of Directed Acyclic Graphs\\ From Discrete Data}
\author{Jiaying Gu\thanks{These two authors contributed equally to this work.}, Fei Fu\footnotemark[1], and Qing Zhou\thanks{
To whom correspondence should be addressed (email: zhou@stat.ucla.edu).}\\
Department of Statistics, University of California, Los Angeles, CA 90095, USA}
\date{}

\maketitle

\begin{abstract}

Bayesian networks, with structure given by a directed acyclic graph (DAG), are a popular class of graphical models. 
However, learning Bayesian networks from discrete or categorical data is particularly
challenging, due to the large parameter space and the difficulty in searching for a sparse structure.
In this article, we develop a maximum penalized likelihood method to tackle this problem. 
Instead of the commonly used multinomial distribution,
we model the conditional distribution of a node given its parents by multi-logit regression, in which
an edge is parameterized by a set of coefficient vectors with dummy variables encoding the levels of a node.
To obtain a sparse DAG, a group norm penalty is employed, and
a blockwise coordinate descent algorithm is developed to maximize the penalized likelihood subject to the  
acyclicity constraint of a DAG. 
When interventional data are available, our method constructs
a causal network, in which a directed edge represents a causal relation.
We apply our method to various simulated and real
data sets. The results show that our method is 
very competitive, compared to many existing methods, in DAG estimation 
from both interventional and high-dimensional 
observational data.


\end{abstract}

\section{Introduction}
\label{sec:dintro}

Bayesian networks are a class of graphical models whose structure
implies conditional independence relationships among a set of
random variables. It is graphically represented by a directed acyclic graph
(DAG). Recent years have seen its popularity in the biological and medical
sciences for inferring gene regulatory networks and cellular
networks, partially attributed to the fact that it can be used for
causal inference. Learning the structure of these biological networks from data is a key to understanding their
functions. Most methods that have been proposed for structure learning of DAGs fall
into two categories. 

The first category encompasses the
constraint-based methods that rely on a set of conditional
independence tests. The tests are used to examine the existence of edges between nodes. In practice the assumptions behind these methods can be strong, which constitutes the main drawback of these methods. The PC algorithm proposed by \citet{Spirtes93} and the MMPC algorithm by \cite{tsamardinos2006max} are two well-known examples. 
The second category includes score-based methods 
whose goal is to search for a DAG that maximizes certain scoring
function. The scoring functions that have been employed include
several Bayesian Dirichlet metrics \citep{Buntine91, Cooper92,
  Heckerman95}, Bayesian information criterion \citep{Chickering97},
minimum description length \citep{Bouckaert93prob,
  Suzuki93, Bouckaert94prob, Lam94}, entropy \citep{Herskovits90}, et
cetera. There are also Monte Carlo methods \citep{Ellis08, Zhou11}
which draw a sample of DAGs from a posterior distribution.
Other recent developments on score-based methods include the work of \cite{scutari2016empirical}, which proposed a posterior score function with an uniform prior for discrete data.
Additionally, there are hybrid methods which combine the two approaches. The idea is to narrow the search space first using a constraint-based method, and then use a score-based method to learn the DAG structure \citep{tsamardinos2006max, gamez2011learning}.

With the rising interest in sparse statistical
modeling, score-based methods seem particularly attractive since
various sparse regularization techniques are potentially applicable. Assuming a given natural ordering among the nodes, \cite{Shojaie10} decomposed
DAG estimation into a sequence of lasso regression problems. 
\cite{schmidt2006lassoordersearch} and \cite{schmidt2007learning} exploited from a computational perspective the idea of using $\ell_1$ regularization to learn the structure of DAGs. 
\cite{Fu13} developed an $\ell_1$-penalized likelihood approach to 
 learn the structure of sparse DAGs from Gaussian data without assuming a given ordering.
This method has been further generalized to the use of concave penalties
by \cite{Aragam14}.  
\cite{han2016estimation} proposed a two-stage adaptive lasso approach for structure estimation of DAGs. 
There are also theoretical developments on $\ell_0$-penalized estimation
of high-dimensional DAGs under a multivariate Gaussian model \citep{geer2013}.

Despite the recent fast developments on sparse regularization methods for learning Gaussian DAGs,
a generalization to discrete data is highly nontrivial.
First, each node
now represents a factor coded by a group of dummy variables. In order to select
a group  of dummy variables together, we need to use a group norm penalty instead of penalizing
individual coefficients. 
Second, the log-likelihood function for categorical data has more parameters, and development of an algorithm to maximize the penalized log-likelihood becomes much more challenging.  
In this paper, we propose a principled generalization of the 
penalized likelihood methodology in our previous work \citep{Fu13} to estimate
sparse DAGs from categorical data without knowing the ordering among variables.
To reduce the parameter space, we use a multi-logit regression to model the conditional distributions
in a discrete Bayesian network. A blockwise coordinate descent (CD) algorithm is developed, which
may take both observational and interventional data. 
Through extensive comparisons, we demonstrate that our method can outperform many competitors 
in learning discrete Bayesian networks from interventional data or from high-dimensional ($p>n$) observational data. Our algorithm has been implemented in the R package, \texttt{discretecdAlgorithm}, available on CRAN.

The remainder of this paper is organized as follows.
Section~\ref{sec:dBN} describes the proposed multi-logit regression model and
formulates the structure learning problem. Section~\ref{sec:dCDalgo} develops in detail the 
blockwise CD algorithm for learning discrete Bayesian networks. 
Section~\ref{sec:dsimul} reports numerical results of our
method on different types of simulated networks for both interventional and 
observational data,
and Section~\ref{sec:realNetworks} presents results on
real networks. Sections \ref{sec:dsimul} and \ref{sec:realNetworks} also include extensive comparisons with other competing methods.
The paper is concluded with a discussion in Section~\ref{sec:dsummary}. 
In the Appendix, we
establish some asymptotic properties of our penalized DAG estimator.

\section{Problem Formulation}\label{sec:dBN}

\subsection{Discrete Bayesian networks}\label{sec:dBayesNet}

The structure of a Bayesian network for $p$ random variables
$X_1,\ldots,X_p$ is given by a DAG
$\mathcal{G}=(V,E)$. The set of nodes $V=\left\{ 1,\ldots,p \right\}$
represents the set of random variables $\left\{ X_1,\ldots,X_p \right\}$, 
and the set of edges is given by
$E = \left\{(j,i) \in V \times V: j \rightarrow i\right\}$, where $j\rightarrow i$ is a directed edge in  $\mathcal{G}$.
Given the structure of $\mathcal{G}$, the joint probability density (mass) function of
$\left( X_1,\ldots,X_p \right)$ can be factorized as
\begin{equation}\label{eq:BNdef}
p(x_1,\ldots,x_p)=\prod_{i=1}^{p}p(x_i|\Uppi_i^{\mathcal{G}}),
\end{equation} 
where  $\Uppi_i^{\mathcal{G}} = \left\{j \in V: (j, i) \in E \right\}$	
is called the set of parents of $X_i$ and $p(x_i|\Uppi_i^{\mathcal{G}})$ denotes
the conditional density of $X_i$ given $\Uppi_i^{\mathcal{G}}$.
Throughout the paper, we use $j$ and $X_j$ interchangeably.

Given a joint distribution, there may exist multiple factorizations of the form in
(\ref{eq:BNdef}), leading to different DAGs. 
The DAGs encoding the same set of conditional independence relations
form an equivalence class.  
All DAGs in the same equivalence class have the same skeleton and v-structures. 
Here, a v-structure is a triplet $\{i, j, k\} \subset V$ of the form $i \rightarrow k \leftarrow j$, 
while $i$ and $j$ are not directly connected.
However, when used for causal inference, equivalent DAGs do
not have the same causal interpretation and can be differentiated
based on experimental data. 
There are methods for learning causal DAGs
from a mix of observational and experimental data \citep{cooper1999causal,meganck2006learning,Ellis08,hauser2015jointly}
and related work on inferring gene networks from perturbed expression data 
\citep{pe2001inferring,pournara2004reconstruction, shojaie2014inferring}.

We describe briefly how the joint
distribution of a Bayesian network can be
modified to incorporate experimental data. For a detailed account of
causal inference using Bayesian networks, please refer to
\citet{Pearl00} and references therein. Assuming $X_i$, 
$i\in \mathcal{M}\subset\{1,\ldots,p\}$, is under experimental intervention, the joint density
in (\ref{eq:BNdef}) becomes
\begin{equation}\label{eq:BNint}
p(x_1,\ldots,x_p)=\prod_{i \notin \mathcal{M}}p(x_i|\Uppi_i^{\mathcal{G}})\prod_{i \in \mathcal{M}}p(x_i|\bullet),
\end{equation}
where $p(x_i|\bullet)$ specifies the distribution of $X_i$ under
intervention. Experimental data generated from
$\mathcal{G}$ can therefore be considered as being generated
from the DAG $\mathcal{G}'$ obtained by removing all directed edges in $\mathcal{G}$
pointing to the variables under intervention. 
It should be noted that \eqref{eq:BNint} also applies to observational data for which
$\mathcal{M}$ is simply empty, and in this case, \eqref{eq:BNint} reduces to \eqref{eq:BNdef}.  Hereafter, we develop our method under the assumption
that part of the data are generated under experimental intervention, while regarding purely observational
data as the special case of $\mathcal{M}=\varnothing$.

In a discrete Bayesian network, each variable $X_i$ is considered a
factor with $r_i$ levels, indexed by $\left\{ 1, \ldots, r_i \right\}$. The set
of its parents $\Uppi_i^{\mathcal{G}}$ has
a total of $q_i = \prod_{j \in \Uppi_i^{\mathcal{G}}} r_j$ possible joint
states $\left\{ \boldsymbol{\pi}_k: k = 1, \ldots, q_i \right\}$. 
Let $\Uptheta_{ijk}=P\left(X_i = j \mid \Uppi_i^{\mathcal{G}}=\boldsymbol{\pi}_k \right)$.
A discrete Bayesian network $\mathcal{G}$ may be 
parameterized by $\boldsymbol{\Uptheta} = \{\Uptheta_{ijk} \geq 0:
\sum_j\Uptheta_{ijk}=1 \}$ via a product multinomial model given the graph structure. 
The number of parameters in this product multinomial model is
\begin{align*}
{N(\boldsymbol{\Uptheta})} =
\sum_{i=1}^p r_iq_i=\sum_{i=1}^p r_i\prod_{j \in \Uppi_i^{\mathcal{G}}} r_j. 
\end{align*}
If we assume that each variable has $\mathcal{O}(r)$ levels, then
\begin{equation}\label{eq:multirate}
{N(\boldsymbol{\Uptheta})} = \mathcal{O}\left(\sum_{i=1}^p r^{1+\abs{\Uppi_i^{\mathcal{G}}}}\right), 
\end{equation}
which grows exponentially as the size of the parent set
$\abs{\Uppi_i^{\mathcal{G}}}$ increases. To reduce the number of free parameters, we
propose a multi-logit model for discrete Bayesian networks under
which development of a penalized likelihood method is 
straightforward. For the same DAG structure,
the number of parameters can be much smaller compared to the product
multinomial model.  

\subsection{A multi-logit model}\label{sec:dplogL}

We encode the $r_i$ levels of $X_i$, $i=1,\ldots,p$, by a group of $d_i=r_i-1$ dummy variables.
Let $\mathbf{x}_{i} \in\{0,1\}^{d_i}$ be the group of dummy variables for $X_i$
and $\mathbf{x}=(1, \mathbf{x}_{1}, \ldots,\mathbf{x}_{p})$ be a $d$-vector, where $d=1+\sum_{i=1}^pd_i$. 
For a discrete Bayesian network $\mathcal{G}$, we model the
conditional distribution $[X_j|\Uppi_j^{\mathcal{G}}]$, $j = 1,\ldots,p$, by the following multi-logit regression model
\begin{eqnarray}\label{eq:mlogit}
P(X_j=\ell \mid \Uppi_j^{\mathcal{G}})&=& \dfrac{\exp(\beta_{j\ell0}+
\sum_{i=1}^p\mathbf{x}_{i}^T\boldsymbol{\beta}_{j\ell i})}{\sum_{m=1}^{r_j} \exp(\beta_{jm0}
+\sum_{i=1}^p\mathbf{x}_{i}^T\boldsymbol{\beta}_{jmi})} \nonumber\\
&=& \dfrac{\exp(\mathbf{x}^T\boldsymbol{\beta}_{j\ell\cdot})}
{\sum_{m=1}^{r_j}\exp(\mathbf{x}^T\boldsymbol{\beta}_{jm\cdot})}\defi p_{j\ell}(\mathbf{x}),
\end{eqnarray}
for $\ell=1, \ldots, r_j$, where $\beta_{j\ell 0}$ is the
intercept, $\boldsymbol{\beta}_{j\ell i} \in \mathbb{R}^{d_i}$ is the
coefficient vector for $X_i$ to predict the
$\ell^{\text{th}}$ level of $X_j$, and 
$\boldsymbol{\beta}_{j\ell\cdot}=\text{vec}(\beta_{j\ell 0}, \boldsymbol{\beta}_{j\ell 1}, \ldots,\boldsymbol{\beta}_{j\ell p}) \in\mathbb{R}^d$.
Note that in (\ref{eq:mlogit}), 
$\boldsymbol{\beta}_{j\ell i}=\mathbf{0}$ for all $\ell$ if 
$i \notin \Uppi_j^{\mathcal{G}}$. 
Thus, our model indeed defines a joint distribution for $X_1, \ldots, X_p$ which factorizes according to the DAG $\mathcal{G}$.
We choose to use a symmetric form of the multi-logit model
here, as was done in \citet{Zhu04} and \citet{Friedman10}. To make this model
identifiable, we impose the following
constraints on the intercepts
\begin{equation}\label{eq:constraint}
\beta_{j1 0} = 0, \quad j=1,\ldots,p.
\end{equation}
The nonidentifiability of other parameters can be resolved via
regularization as demonstrated by \citet{Friedman10}. The particular
form of regularization we use leads to the following constraints
\begin{equation}\label{eq:constraint2}
\sum_{m=1}^{r_j}\boldsymbol{\beta}_{jm i}=\mathbf{0}, \quad \forall\, i,j=1,\ldots,p.
\end{equation}
Let $\boldsymbol{\beta} =(\boldsymbol{\beta}_{jmi})$, which is a four-way array, 
denote all the parameters.
Given the structure of $\mathcal{G}$, the number of free parameters is
\begin{equation}
N(\boldsymbol{\beta})=\sum_{j=1}^p \left[(r_j-1)+r_j\sum_{i \in  \Uppi_j^{\mathcal{G}}}d_i\right].
\end{equation} 
If we further assume that $r_i = \mathcal{O}(r)$ for all $i$, then
\begin{equation}\label{eq:multiLogitBound}
N(\boldsymbol{\beta}) = \mathcal{O}(r^2)|E| + \mathcal{O}(rp),
\end{equation}
which grows linearly in the total number of edges $|E|$. 
This rate of growth is much
slower than that of the product multinomial model (\ref{eq:multirate}).
Note that these two models are not equivalent.
Numerical comparisons in Section \ref{sec:realNetworks} 
confirm that the proposed multi-logit model often serves as a good approximation to the product multinomial model.

Suppose that we have a data set $\mathcal{X}=(\mathcal{X}_{hi})_{n \times p}$ 
generated from a causal discrete Bayesian network $\mathcal{G}$, where
$\mathcal{X}_{hi}$ is the 
level of $X_i$ in the $h^{\text{th}}$ data point, $h=1,\ldots,n$, coded by dummy variables $\mathbf{x}_{h,i}\in\{0,1\}^{d_i}$.
Let $\mathcal{I}_j$ be the index set of rows where $X_j$ is under intervention and $\mathcal{O}_j=\left\{1,\ldots,n\right\} \backslash \mathcal{I}_j$ be the index set of
rows in which $X_j$ is observational.
Note that $\mathcal{I}_j$ are not necessarily mutually exclusive, which means there can be more than one node under intervention for a data point.
Under the multi-logit model (\ref{eq:mlogit}), the log-likelihood function $\ell(\boldsymbol{\beta})$
can be written according to the factorization (\ref{eq:BNint}) as
\begin{eqnarray}\label{eq:dlogL}
\ell(\boldsymbol{\beta}) 
&=& \sum_{j=1}^p \sum_{h \in \mathcal{O}_j} \log\left[p(\mathcal{X}_{hj} |
\mathbf{x}_{h,i}, i \in \Uppi_j^{\mathcal{G}})\right] \nonumber \\
& = & \sum_{j=1}^p \sum_{h \in \mathcal{O}_j} \left[\sum_{\ell=1}^{r_j}y_{hj\ell}
    \mathbf{x}_h^T\boldsymbol{\beta}_{j\ell \cdot} - \log\left\{\sum_{m=1}^{r_j}\exp(\mathbf{x}_h^T\boldsymbol{\beta}_{jm\cdot})\right\}\right],
\end{eqnarray}
where $y_{hj\ell}=I(\mathcal{X}_{hj}=\ell)$ are indicator
variables and $\bmbeta_{j\ell k}=\bfzr$ for $k \notin \Uppi_j^{\calG}$.

\begin{remark}\label{rem:ObsEInt}
Although we have assumed the availability of experimental data, it is easy to see
that the log-likelihood \eqref{eq:dlogL} applies to observational data as well: If there are no experimental data
for $X_j$, then $\mathcal{O}_j=\{1$ $,\ldots$ $,n\}$ in \eqref{eq:dlogL}.
\end{remark}

\subsection{Group norm penalty}

Define $\boldsymbol{\beta}_{j\cdot i}
=\text{vec}(\boldsymbol{\beta}_{j1i}, \ldots,
\boldsymbol{\beta}_{jr_j i}) \in \mathbb{R}^{d_ir_j}$ to be
the vector of coefficients representing the influence of $X_i$ on
$X_j$ and $\boldsymbol{\beta}_{j\cdot 0} = (\beta_{j10},
\ldots, \beta_{jr_j0}) \in \mathbb{R}^{r_j}$ to be the vector of intercepts for predicting
$X_j$. 
The structure of $\mathcal{G}$
is coded by the sparsity of $\boldsymbol{\beta}_{j\cdot i}$ as
\begin{equation}\label{eq:sld}
\boldsymbol{\beta}_{j\cdot i}=\mathbf{0} \qquad\Longleftrightarrow\qquad 
i \notin \Uppi_j^{\mathcal{G}}.
\end{equation}
In order to learn a sparse DAG from data, we estimate $\boldsymbol{\beta}$ via
a penalized likelihood approach. It can
be seen from (\ref{eq:sld}) that, for discrete Bayesian networks, the
set of parents of $X_j$ is given by the set 
$\{i: \boldsymbol{\beta}_{j\cdot i} \ne \mathbf{0}\}$. 
The regular $\ell_1$ penalty is inappropriate for this
purpose since it penalizes each component of $\boldsymbol{\beta}$
separately. We instead penalize the vector $\boldsymbol{\beta}_{j\cdot i}\in\mathbb{R}^{d_ir_j}$ 
as a whole to obtain a sparse DAG via the use of a group norm penalty.  
Group norm penalties have been used in the group lasso and its generalizations \citep{Yuan06,Meier08}. 
Let $\mathcal{G}_{\boldsymbol{\beta}}$
denote the graph induced by $\boldsymbol{\beta}$ so that
$\Uppi_j^{\mathcal{G}_{\boldsymbol{\beta}}}=\{i: \boldsymbol{\beta}_{j\cdot i} \ne \mathbf{0}\}$ for $j=1,\ldots,p$. 
We define our group norm penalized estimator
for a discrete Bayesian network by the following optimization program:
\begin{eqnarray}\label{eq:dobj1}
f_{\lambda}(\boldsymbol{\beta}) 
&\defi& -\ell(\boldsymbol{\beta})+\lambda\sum_{j=1}^p\sum_{i=1}^p \norm{\boldsymbol{\beta}_{j\cdot i}}_2,\\
\hat{\boldsymbol{\beta}}_{\lambda} \label{eq:argminBeta}
&=& \mathop{\arg\min}_{\boldsymbol{\beta}: \mathcal{G}_{\boldsymbol{\beta}}
  \text{ is a DAG}} f_{\lambda}(\boldsymbol{\beta}),
\end{eqnarray}
where 
$\lambda>0$ is a tuning parameter. 
See Section~\ref{sec:path} for choosing the parameter $\lambda$.  
The feasible set of (\ref{eq:argminBeta}) is a DAG space, which imposes a highly nonconvex constraint.
This is a major challenge for our optimization algorithm.
Hereafter, we call $\bmbeta_{j\cdot i}$ a (component) group of $\bmbeta$.


\section{Algorithm}\label{sec:dCDalgo}

Structure learning for discrete Bayesian networks is computationally
demanding because of the
nonlinear nature of the multi-logit model (\ref{eq:mlogit}). We develop in this section a blockwise coordinate
descent algorithm to solve (\ref{eq:argminBeta}). Coordinate descent algorithms have been proved successful in
various settings \citep{Fu98, Friedman07, Wu08} and their
implementations are relatively straightforward.

\subsection{Single coordinate descent step}\label{sec:Newton}

We first consider minimizing $f_{\lambda}(\boldsymbol{\beta})$ (\ref{eq:dobj1}) with respect to
$\boldsymbol{\beta}_{j\cdot  i}$ while holding all the
other parameters constant. We define
\begin{align}\label{eq:CDstep}
f_{\lambda,j}(\boldsymbol{\beta}_{j\cdot \cdot}) 
&= -\sum_{h \in \mathcal{O}_j} \left[
  \sum_{\ell=1}^{r_j}y_{hj\ell}
    \mathbf{x}_h^T\boldsymbol{\beta}_{j\ell \cdot} -\log\left\{\sum_{m=1}^{r_j}\exp\left(\mathbf{x}_h^T\boldsymbol{\beta}_{jm\cdot}\right)\right\}
  \right] 
  +\lambda\sum_{k=1}^p \norm{\boldsymbol{\beta}_{j\cdot k}}_2 \nonumber \\
& \defi -\ell_j\left(\boldsymbol{\beta}_{j\cdot \cdot}\right)
+\lambda\sum_{k=1}^p \norm{\boldsymbol{\beta}_{j\cdot k}}_2 ,
\end{align}
where $\boldsymbol{\beta}_{j\cdot \cdot}=(\boldsymbol{\beta}_{j\cdot 0},
\boldsymbol{\beta}_{j\cdot 1},\ldots, \boldsymbol{\beta}_{j\cdot p})$. 
Considering the
problem of minimizing $f_{\lambda,j}(\cdot)$ over $\boldsymbol{\beta}_{j\cdot i}$, 
we write $f_{\lambda,j}$ and $\ell_j$ as $f_{\lambda,j}(\boldsymbol{\beta}_{j\cdot i})$ and $\ell_j(\boldsymbol{\beta}_{j\cdot i})$, respectively. 

Following the approach of
\citet{Tseng09} and \citet{Meier08}, we form a
quadratic approximation to $\ell_j(\boldsymbol{\beta}_{j\cdot i})$
using the second-order Taylor expansion at $\boldsymbol{\beta}_{j\cdot i}^{(t)}$,
the current value of $\bmbeta_{j\cdot i}$. 
Adding the penalty, the quadratic approximation is
\begin{eqnarray}\label{eq: Newtonstep}
Q_{\lambda,j}^{(t)}(\boldsymbol{\beta}_{j\cdot i})
 &= & -\left\{ \left(\boldsymbol{\beta}_{j\cdot i}-\boldsymbol{\beta}_{j\cdot
    i}^{(t)}\right)^T\nabla{\ell}_j\left(\boldsymbol{\beta}_{j\cdot i}^{(t)}\right)\right. \nonumber \\
 & &\left. +\dfrac{1}{2}\left(\boldsymbol{\beta}_{j\cdot i}-\boldsymbol{\beta}_{j\cdot
    i}^{(t)}\right)^T\bfH_{ji}^{(t)}\left(\boldsymbol{\beta}_{j\cdot i}-\boldsymbol{\beta}_{j\cdot
    i}^{(t)}\right)\right\} +\lambda \norm{\boldsymbol{\beta}_{j\cdot i}}_2,
\end{eqnarray}
up to an additive term that does not depend on $\bmbeta_{j\cdot i}$. The gradient of the log-likelihood function $\ell_j(\cdot)$ is
\begin{equation}\label{eq:dgradient}
\nabla{\ell}_j(\boldsymbol{\beta}_{j\cdot i}^{(t)})=
\sum_{h \in \mathcal{O}_j}
\begin{pmatrix}
\left(y_{hj1}-p_{j1}^{(t)}(\mathbf{x}_h)\right)\mathbf{x}_{h,i}\\ \vdots\\ \left(y_{hjr_j}-p_{jr_j}^{(t)}(\mathbf{x}_h)\right)\mathbf{x}_{h,i}
\end{pmatrix},
\end{equation}
where $p_{j\ell}^{(t)}(\mathbf{x})$ , $\ell =1,\ldots, r_j$, defined in \eqref{eq:mlogit} are
evaluated at the current parameter values. To give a reasonable quadratic approximation, we use a negative definite matrix
$\bfH_{ji}^{(t)}=h_{ji}^{(t)}\mathbf{I}_{d_ir_j}$ in (\ref{eq: Newtonstep}) to approximate the Hessian of $\ell_j(\cdot)$, 
where the scalar $h_{ji}^{(t)}<0$ and
$\mathbf{I}_{d_ir_j}$ is the identity matrix of size $d_ir_j \times d_i r_j$. 
We choose 
\begin{equation}\label{eq:hessian}
h_{ji}^{(t)} =h_{ji}(\bmbeta^{(t)})\defi -\max\{\text{diag}(-\bfH_{\ell_j}(\boldsymbol{\beta}_{j\cdot i}^{(t)})), b\},
\end{equation}
where $\bfH_{\ell_j}$ is the Hessian of the log-likelihood function $\ell_j(\cdot)$ and $b$ is a small positive number 
used as a lower bound to help convergence. Note that it is not necessary to recompute $h_{ji}^{(t)}$ every iteration (\citealp{Meier08}). See Section~\ref{sec:path} for more details.

It is not difficult to show the following proposition, 
which is a direct consequence of the Karush-Kuhn-Tucker (KKT) conditions for minimizing \eqref{eq: Newtonstep}.
\begin{proposition}\label{prop:Newton}
Let $\bfH_{ji}^{(t)}=h_{ji}^{(t)}\mathbf{I}_{d_ir_j}$ for some scalar
$h_{ji}^{(t)}<0$ and 
$\boldsymbol{d}_{ji}^{(t)}=\nabla{\ell}_j(\boldsymbol{\beta}_{j\cdot i}^{(t)})-h_{ji}^{(t)}\boldsymbol{\beta}_{j\cdot i}^{(t)}$.
Then, the minimizer of
$Q_{\lambda,j}^{(t)}(\boldsymbol{\beta}_{j\cdot i})$ in \eqref{eq: Newtonstep} is
\begin{align}\label{eq:edgeupdate}
\boldsymbol{\bar{\beta}}_{j\cdot i}^{(t)}=\left\{\begin{array}{ll}
\mathbf{0} & \text{ if } \norm{\boldsymbol{d}_{ji}^{(t)}}_2 \leq \lambda, \\
-\dfrac{1}{h_{ji}^{(t)}}\left[
  1 - \dfrac{\lambda}{\norm{\boldsymbol{d}_{ji}^{(t)}}_2} \right] \boldsymbol{d}_{ji}^{(t)}
  & \text{ otherwise.}
\end{array}\right.
\end{align}
\end{proposition}

In order to achieve sufficient descent, an inexact line search by the Armijo rule is performed 
when $\boldsymbol{\bar{\beta}}_{j\cdot i}^{(t)} \neq \boldsymbol{\beta}_{j\cdot i}^{(t)}$, following the procedure in \citet{Meier08}.
Put
$\boldsymbol{s}_{ji}^{(t)} = \boldsymbol{\bar{\beta}}_{j\cdot i}^{(t)} - \boldsymbol{\beta}_{j\cdot i}^{(t)}$,
and let $\Delta^{(t)}$ be the change in $f_{\lambda,j}$ when the log-likelihood 
is linearized at $\bmbeta_{j\cdot i}^{(t)}$, i.e.,
\begin{equation*}
\Delta^{(t)} = - (\boldsymbol{s}_{ji}^{(t)})^T\nabla{\ell}_j(\boldsymbol{\beta}_{j\cdot i}^{(t)}) + \lambda\|\boldsymbol{\bar{\beta}}_{j\cdot i}^{(t)}\|_2 - \lambda \|\boldsymbol{\beta}_{j\cdot i}^{(t)}\|_2.
\end{equation*}
Pick $\eta, \delta \in (0, 1)$ and $\alpha_0 > 0$, and let 
$\alpha^{(t)}$ be the largest value in the sequence $\{\alpha_0\eta^k\}_{k\geq 0}$ such that 
\begin{equation*}
f_{\lambda,j}(\bmbeta_{j\cdot i}^{(t)}+\alpha^{(t)}\boldsymbol{s}_{ji}^{(t)})
\leq f_{\lambda,j}(\bmbeta_{j\cdot i}^{(t)}) + \delta\alpha^{(t)}\Delta^{(t)}.
\end{equation*}
Then set
\begin{equation} \label{eq:rule}
\boldsymbol{\beta}_{j\cdot i}^{(t+1)} = \boldsymbol{\beta}_{j\cdot i}^{(t)}+\alpha^{(t)}\boldsymbol{s}_{ji}^{(t)},
\end{equation}
which completes one iteration for updating $\bmbeta_{j\cdot i}$.
In our implementation, we choose $\eta= 0.5$, $\delta=0.1$ and $\alpha_0=1$ following the suggestion by \cite{Meier08}.

It follows from Proposition~\ref{prop:Newton} with $\lambda=0$ that for the unpenalized intercepts, 
\begin{equation}\label{eq:interceptupdate}
\boldsymbol{\beta}_{j\cdot 0}^{(t+1)} 
= \boldsymbol{\bar{\beta}}_{j\cdot 0}^{(t)}= -\boldsymbol{d}_{j0}^{(t)}/h_{j0}^{(t)}.
\end{equation}
In addition, some of the parameters are always constrained to zero, e.g., 
$\boldsymbol{\beta}_{j\cdot j}$ and $\beta_{j1 0}$ for all $j$.

\subsection{Blockwise coordinate descent}\label{sec:dCDAlgo}

Our CD algorithm consists of two layers of iterations. In the outer loop, we cycle through all pairs of nodes
to update the active set of edges, including their directions. In the inner loop, we only cycle through the
active edge set to update the parameter values.

We first describe the outer loop. Due to the acyclicity constraint in (\ref{eq:argminBeta}), we know \emph{a priori} that
$\boldsymbol{\beta}_{i\cdot j}$ and $\boldsymbol{\beta}_{j\cdot i}$
cannot simultaneously be nonzero for $i\ne j$. This suggests
performing the minimization in blocks, minimizing over
$\{\boldsymbol{\beta}_{i\cdot j},\boldsymbol{\beta}_{j\cdot i}\}$
simultaneously. 
In order to enforce acyclicity, we use a simple heuristic \citep{Fu13}: For each block $\{\boldsymbol{\beta}_{i\cdot j},\boldsymbol{\beta}_{j\cdot i}\}$, we check if adding an edge from $i\to j$ induces a cycle in the estimated DAG. If so, we set $\boldsymbol{\beta}_{j\cdot i}=\bfzr$ and minimize with respect to $\boldsymbol{\beta}_{i\cdot j}$. Alternatively, if the edge $j\to i$ induces a cycle, we set $\boldsymbol{\beta}_{i\cdot j}=\bfzr$ and minimize with respect to $\boldsymbol{\beta}_{j\cdot i}$. If neither edge induces a cycle, we minimize over both parameters simultaneously. The cycle check is implemented by a breath-first search algorithm. 
We outline below (Algorithm~\ref{algo:dCD}) the complete blockwise CD 
algorithm for discrete Bayesian networks. In the algorithm, $\boldsymbol{\beta}_{j\cdot i}
\Leftarrow \bfzr$ is used to indicate that 
$\boldsymbol{\beta}_{j\cdot i}$ must be set to zero due to the acyclicity
constraint given the current estimates of the other parameters. 
Minimization of $f_{\lambda, j}(\cdot)$ with respect to
$\boldsymbol{\beta}_{j\cdot i}$ is done with the single CD step with line search (\ref{eq:rule}).

\begin{algorithm}[ht]
\caption{CD algorithm for estimating discrete Bayesian networks}
\label{algo:dCD}
\begin{algorithmic}[1]
\State Initialize $\boldsymbol{\beta}$ such that
$\mathcal{G}_{\boldsymbol{\beta}}$ is acyclic
       \For {$i = 1, \ldots, p-1$}\label{line:begin}
             \For {$j = i+1, \ldots, p$}
                \If {$\boldsymbol{\beta}_{j\cdot i} \Leftarrow \mathbf{0}$}
                   \State $\boldsymbol{\beta}_{i\cdot j} \leftarrow \mathop{\arg\min}_{
    \boldsymbol{\beta}_{i\cdot j}}f_{\lambda, i}(\cdot)$, \;
  $\boldsymbol{\beta}_{j\cdot i} \leftarrow \mathbf{0}$
                \ElsIf {$\boldsymbol{\beta}_{i\cdot j} \Leftarrow \mathbf{0}$}
                   \State $\boldsymbol{\beta}_{i\cdot j} \leftarrow
                   \mathbf{0}$, \; $\boldsymbol{\beta}_{j\cdot i} \leftarrow \mathop{\arg\min}_{
    \boldsymbol{\beta}_{j\cdot i}}f_{\lambda, j}(\cdot)$ 
                \Else
                   \State $S_1 \leftarrow \mathop{\min}_{
    \boldsymbol{\beta}_{i\cdot j}}f_{\lambda, i}(\cdot) + f_{\lambda,
    j}(\cdot)|_{\boldsymbol{\beta}_{j\cdot i}=\mathbf{0}}$
                   \State $S_2 \leftarrow f_{\lambda,
    i}(\cdot)|_{\boldsymbol{\beta}_{i\cdot j}=\mathbf{0}} + \mathop{\min}_{
    \boldsymbol{\beta}_{j\cdot i}}f_{\lambda, j}(\cdot)$
                   \If {$S_1 \leq S_2$}
                       \State $\boldsymbol{\beta}_{i\cdot j} \leftarrow \mathop{\arg\min}_{
    \boldsymbol{\beta}_{i\cdot j}}f_{\lambda, i}(\cdot)$, \;
  $\boldsymbol{\beta}_{j\cdot i} \leftarrow \mathbf{0}$
                   \Else
                       \State $\boldsymbol{\beta}_{i\cdot j} \leftarrow
                   \mathbf{0}$, \; $\boldsymbol{\beta}_{j\cdot i} \leftarrow \mathop{\arg\min}_{
    \boldsymbol{\beta}_{j\cdot i}}f_{\lambda, j}(\cdot)$
                   \EndIf
                \EndIf
             \EndFor
          \EndFor
\State  Update intercepts $\boldsymbol{\beta}_{j\cdot 0}$
 for $j=1,\ldots,p$ 
\State Inner loop given the active edge set (Algorithm~\ref{algo:linesearch}) \label{line:end}
\State  Repeat step \ref{line:begin} to \ref{line:end} until some stopping criterion is met
\end{algorithmic}
\end{algorithm}

Let $\bmbeta^{(t)}$ denote the parameter value after one cycle of the outer loop 
(after line 19 in Algorithm~\ref{algo:dCD}). 
Denote its active edge set by 
$E^{(t)}=\{(i,j):\bmbeta^{(t)}_{j\cdot i}\ne \mathbf{0}\}$.
The inner loop solves the following problem:
\begin{equation}\label{eq:innerobj}
\min_{\bmbeta} f_{\lambda}(\bmbeta), \;\text{subject to } \supp(\bmbeta)\subset E^{(t)},
\end{equation}
where $f_{\lambda}$ is defined in \eqref{eq:dobj1}. We use $\boldsymbol{\beta}^{(t)}$ as the initial value 
and cycle through $\bmbeta_{j\cdot i}$ for $(i,j)\in E^{(t)}$. In particular, the direction of an edge will not be reversed
but edges may be deleted if their parameters $\bmbeta_{j\cdot i}$ are updated to zero.
See Algorithm~\ref{algo:linesearch} for an outline of the inner loop.

\begin{algorithm}[ht]
\caption{Inner loop}
\label{algo:linesearch}
\begin{algorithmic}[1]
\State Input $E^{(t)}$ and initialize $\bmbeta=\bmbeta^{(t)}$
	\For {$(i, j) \in E^{(t)}$}\label{line:searchbegin}
		\State $\boldsymbol{\beta}_{j\cdot i} \leftarrow \mathop{\arg\min}_{
    \boldsymbol{\beta}_{j\cdot i}}f_{\lambda, j}(\cdot)$
	\EndFor\label{line:searchend}
\State Repeat step \ref{line:searchbegin} to step \ref{line:searchend} until convergence.
\end{algorithmic}
\end{algorithm}

By construction, $E^{(t)}$ satisfies the acyclicity constraint and thus the feasible region
in \eqref{eq:innerobj} is simply a Euclidean space. Since $f_{\lambda}$ itself is convex,
the CD algorithm for the inner loop has nice convergence properties.
In analogy to Proposition 2 in \citet{Meier08}, we arrive at the following convergence result.

\begin{proposition}\label{prop:convergence}
Suppose that the sequence $\{\bmbeta^{(k)}\}$ is generated by the inner loop.
If the matrix $\bfH_{ji}^{(k)}$ is chosen according to \eqref{eq:hessian}, then every limit point of the
sequence $\{\bmbeta^{(k)}\}$ is a minimizer of problem \eqref{eq:innerobj}.
\end{proposition}

However, since the search space for the outer loop is the full DAG space, which is highly nonconvex, 
rigorous theory on its convergence is yet to be established. 
Therefore, a practical stopping criterion is employed. After the convergence of an inner loop,
we obtain the current active set. If one more iteration of the outer loop does not change the active set,
we then stop Algorithm \ref{algo:dCD}. On the other hand, we also set a maximum number of iterations
for the outer loop.
For all the examples we have tested, 
our CD algorithm has shown no problem in convergence. 
This empirical observation is in line with recent theoretical work by 
\cite{lee2016gradient}
who have established that gradient descent converges to a local minimizer of 
a nonconvex objective function for almost all initial values and have suggested similar behavior for coordinate descent.

\subsection{Solution path}\label{sec:path}

We use Algorithm~\ref{algo:dCD} to compute
$\hat{\boldsymbol{\beta}}_{\lambda}$ (\ref{eq:argminBeta}) over a grid of
$J$ values for the tuning parameter, $\lambda_1 > \ldots > \lambda_J  > 0$, where at
$\lambda_1$ every parameter other than the intercepts is estimated as
zero. It follows from the KKT conditions for (\ref{eq:dobj1}) that
\begin{eqnarray}\label{eq:dlambda_max}
\lambda_1 
&=& \max_{1\leq i,j \leq p}
\norm{\nabla{\ell}_j(\boldsymbol{\beta}_{j\cdot
    i}) |_{\boldsymbol{\beta}_{j\cdot
    i}=\mathbf{0}}}_2,
\end{eqnarray}
in which $\boldsymbol{\beta}_{j\cdot 0}$ is set to the MLE of the intercept assuming 
all $\boldsymbol{\beta}_{j\cdot i}$, $i=1,\ldots,p$, are zero.

The solution $\hat{\boldsymbol{\beta}}_{\lambda_m}$ is used as a warm start for
estimating $\hat{\boldsymbol{\beta}}_{\lambda_{m+1}}$, $m=1, \ldots,
J-1$. To save computational time, we set $h_{ji}^{(t)}=h_{ji}(\hat{\bmbeta}_{\lambda_{m}})$ \eqref{eq:hessian}
in the CD algorithm for $\hat{\bmbeta}_{\lambda_{m+1}}$, instead of updating $h_{ji}^{(t)}$ every iteration.

Traditional model selection criteria such as
BIC do not work well for the purpose
of estimating DAGs from data.
In our simulation results, the hill-climbing (HC) algorithm \citep{gamez2011learning}, which uses BIC as the scoring function, always selects too many edges. There are also numerical studies \citep{scutari2016empirical} in which BIC tends to select too few edges on a different set of DAGs, showing that BIC could be sensitive and unstable. In order to select a suitable tuning
parameter, we use an empirical model selection criterion 
proposed by \citet{Fu13}. Let $\hat{\mathcal{G}}_{\lambda_m}$ be the DAG induced by
$\hat{\boldsymbol{\beta}}_{\lambda_m}$ and $e_{\lambda_m}$ be the
number of edges in $\hat{\mathcal{G}}_{\lambda_m}$. We reestimate 
$\bmbeta$ by the maximizer
$\bmbeta^{\dag}_{\lambda_m}$ of the log-likelihood $\ell(\bmbeta)$ \eqref{eq:dlogL} 
given $\calG=\hat{\mathcal{G}}_{\lambda_m}$ using the R package \texttt{nnet} \citep{nnet2002}. We define the difference
ratio between two estimated DAGs $\hat{\mathcal{G}}_{\lambda_m}$ and 
$\hat{\mathcal{G}}_{\lambda_{m+1}}$ by
$dr_{(m,m+1)}=\Delta \ell_{(m, m+1)} / \Delta e_{(m,m+1)}$, where $\Delta
\ell_{(m, m+1)} =
\ell(\bmbeta^{\dag}_{\lambda_{m+1}})-\ell(\bmbeta^{\dag}_{\lambda_m})$
and $\Delta e_{(m,m+1)} = e_{\lambda_{m+1}}$ $- e_{\lambda_m} $, if 
$\Delta e_{(m,m+1)}\geq 1$. 
Otherwise, we set $dr_{(m,m+1)}$ $= dr_{(m-1,m+1)}$. The
selected tuning parameter is indexed by
\begin{eqnarray}\label{eq:dempiricalrule}
m^* = \sup \left\{2\leq m \leq J: dr_{(m-1,m)} \geq \alpha \cdot \max \{dr_{(1,2)}, \ldots,
    dr_{(J-1,J)}\} \right\}.
\end{eqnarray}
According to this criterion, an increase in model complexity, measured by the number of predicted edges, 
is accepted only if there is a substantial increase in the log-likelihood. We choose $\alpha=0.3$ for all the results in this work.

\section{Simulation Studies}\label{sec:dsimul}

We evaluate the CD algorithm on simulated data sets. As stated in Remark \ref{rem:ObsEInt}, the log-likelihood \eqref{eq:dlogL} applies to observational data as well. Therefore, we apply the CD algorithm on both interventional data and observational data. In order to assess the accuracy and efficiency of the CD algorithm, we compare it with a few competing methods. For interventional data, we compare our CD algorithm with the PC algorithm \citep{kalisch2007estimating}, the greedy interventional equivalent search (GIES) algorithm \citep{hauser2015jointly} and the equi-energy sampler (EE sampler) \citep{kou2006discussion}.
For observational data we compare the CD algorithm with the hill-climbing (HC) algorithm \citep{gamez2011learning}, 
the max-min hill-climbing (MMHC) algorithm \citep{tsamardinos2006max}, and the PC algorithm.
Among these competitors, the PC algorithm is a constraint-based method, the MMHC is a hybrid method
and the others are all score-based.

Details about data generation and parameter choices will be discussed in Section~\ref{sec:setup}. 
In Section~\ref{sec:ExpwInt}, we compare DAGs estimated from interventional data. Section~\ref{sec:ExpwObs}, on the other hand, presents results on high-dimensional observational data.
The comparison of running times is provided in Section~\ref{sec:timing}. 

\subsection{Experimental setup} \label{sec:setup}

Four types of networks are used to compare the methods: the bipartite graph, 
the scale-free network, the small-world network, and random DAGs. 
In each setting, we consider the combination of three main parameters: $(n, p, s_0)$, 
where $n$ is the sample size, $p$ is the number of nodes, and $s_0$ is the number of true edges. 

We generated bipartite graphs, scale-free networks, and small-world networks 
with the R package \texttt{igraph} \citep{Csardi06}. 
The bipartite graphs were generated by the Erd\H{o}s-R\'{e}nyi model \citep{renyi1959random}.
Each bipartite graph in our datasets had $0.2p$ top nodes, $0.8p$ bottom nodes, and $s_0 = p$ directed edges
from the top to the bottom nodes. The structure of a
scale-free network was generated using the
Barab\'{a}si-Albert model \citep{Barabasi99}. These networks had $s_0 = p-1$
directed edges. The small-world networks were generated
by the Watts-Strogatz model \citep{Watts98}. A graph initially generated
by the model was undirected. To convert it to
a DAG, edge directions were chosen according to a randomly
generated topological sort. 
In this way, a small-world network had $s_0 =2p$ directed edges. 
Random DAGs were sampled using the R package \texttt{pcalg} \citep{kalisch2012causal}, 
and each DAG had $s_0 \approx p$ edges.

In all the simulation studies, each variable was assumed
to be binary, i.e., $r_j=2$ for all $j$. In this case, each group of parameters 
$\bmbeta_{j\cdot i}=(\beta_{j1i},\beta_{j2i}) \in \mathbb{R}^2$.
If $\Uppi_j = \varnothing$,
$X_j$ would be sampled from its two levels with equal
probability. Otherwise, the parameters $\bmbeta_{j\cdot 0}$ and $\bmbeta_{j\cdot i}$, $i \in \Uppi_j$, 
were chosen such that
\begin{equation}\label{eq:LogitPara}
p_{j\ell}(\mathbf{x}_h) = \dfrac{\exp(2\sum_{i \in \Uppi_{j}}
  y_{hi\ell})}{\exp(2\sum_{i \in \Uppi_{j}} y_{hi1}) +
\exp(2\sum_{i \in \Uppi_{j}} y_{hi2})}
\nonumber \\
\end{equation}
for $\ell=1,2$, where $y_{hi\ell}=I(\mathcal{X}_{hi}=\ell)$.
The value of a variable under intervention was randomly fixed to one of its levels regardless of 
its parents.

For each dataset, we input to our CD algorithm a sequence of 40 values of $\lambda$, starting from $\lambda_1$~\eqref{eq:dlambda_max} and ending at $0.01\lambda_1$. Since we assume the graphs are sparse, we stop a solution path when the number of predicted edges exceeds $3p$. Consequently, a sequence of DAGs is estimated and one
of them will be picked by our model selection criterion~\eqref{eq:dempiricalrule} with $\alpha = 0.3$.
To avoid any potential bias in the estimation, we pick a random order to cycle through all the blocks in the outer loop of Algorithm \ref{algo:dCD}.

The EE sampler is used for comparisons on interventional data. 
Its implementation for DAG estimation was done as in \cite{Zhou11}.
We will call it the EE-DAG sampler hereafter.
In each run, we simulate $10$ chains with a $0.1$ chance of equi-energy jumps. 
To obtain an estimated DAG, we threshold the average graph from the target chain (the $10^{\text{th}}$ chain). 
More precisely, a directed edge will be predicted if its posterior inclusion probability is greater than $0.5$.

The HC algorithm is a standard greedy method, while the MMHC algorithm is a hybrid method. For these two algorithms, we use the R package \texttt{bnlearn} \citep{scutari2010learning, Scutari2017learnining}. 
These methods are designed specifically for observational data and thus are compared with our method only on observational data. For the HC algorithm, since the number of predicted edges is too large, we set the maximum number of parents for each node to be 2. 
For the MMHC algorithm we also have the option to limit the number of parents per node. But since this algorithm can estimate a reasonable number of edges, we did not set an upper limit.

The PC algorithm is a popular constraint-based algorithm for learning Bayesian networks, with an efficient 
implementation in the R package \texttt{pcalg} \citep{kalisch2012causal, hauser2012learning}. 
However, it may not produce a DAG for every data set, and instead its output is a completed partially directed acyclic graph (CPDAG), which contains both directed and undirected edges. To make a fair comparison, we distinguish undirected edges from directed ones in our calculation of various performance metrics, with details provided later. The tuning parameter for the PC algorithm is the significance level for conditional independence tests, which is chosen as $\alpha = 0.01$ for all data. 

The GIES algorithm can learn Bayesian networks from a mixture of observational and interventional data,
by searching over the so-called interventional equivalence classes.
However, the implementation of this algorithm in the \texttt{pcalg} package, the only implementation we found, can only take continuous data as input. So we generated continuous data from the simulated discrete data. The detailed procedure is described in the Supplemental Material.

\subsection{Results for interventional data}\label{sec:ExpwInt}

For each type of network, we generated graphs with $p=50$ and $p=100$.  For each node $X_j$, we generated $n_j$ data points where the node is under intervention, so that the sample size $n=\sum_{j=1}^p n_j$ for interventional data. We chose $n_j\in\{1,5\}$ for all $j = 1, \ldots, p$ to
test the performance of the algorithms given different amount of intervention. In particular, when $n_j=1$
we have $n=p$, which lies on the boundary between low- and high-dimensional settings.
In combination, our choices of the data size were $(n, p)$ $\in$ $\{(50, 50)$, $(250, 50)$, $(100, 100)$, $(500, 100)\}$. 
 For each combination of $(n,p)$, we generated 20 data sets.

We compare the DAGs estimated by four algorithms, the CD algorithm, the EE-DAG sampler, 
 the PC algorithm, and the GIES algorithm. 
For an estimated DAG, we distinguish between expected edges, which are estimated edges in the true skeleton with the correct direction, and reversed edges, which are in the true skeleton but with a reversed direction. 
Let P, E, R, and FP denote, respectively, the numbers of predicted edges,
expected edges, reversed edges, and false
positive edges (excluding the reversed ones) in an estimated DAG, and recall that $s_0$ is
the number of edges in the true graph. Then the number of missing edges is $\text{M}=s_0-\text{E}-\text{R}$.
The accuracy of DAG estimation is measured
by the true positive rate (TPR), the false discovery rate (FDR), the structural Hamming distance (SHD) \citep{tsamardinos2006max}, and the Jaccard index (JI), defined as
$\text{TPR} = \text{E}/s_0$,
$\text{FDR} =(\text{R}+\text{FP})/\text{P}$, 
$\text{SHD} =(\text{R}+\text{M}+\text{FP})$, and
$\text{JI} = \text{E}/(\text{P}+s_0-\text{E})$.
Note that SHD was originally defined for CPDAGs, and our definition here measures the Hamming distance between two DAGs.
Both SHD and JI are single performance metrics for DAG estimation. We mark in boldface results with the optimum SHD and JI scores in the subsequent tables.

\begin{table*}[t]
\begin{center}
\caption{Comparison between our CD algorithm and the PC algorithm on simulated interventional data}
\label{tab:cmpPC}       
\resizebox{\linewidth}{!}{%
\begin{tabular}{lclrrrrrrcc}
\hline\noalign{\smallskip}
Graph & $(n, p, s_0)$ & Method & P & E & R & FP & TPR & FDR &  SHD & JI \\
\noalign{\smallskip}\hline\noalign{\smallskip}
Bipartite & {$(100, 100, 100)$} & CD & 98.2 & 63.0 & 18.6 & 16.5 & 0.630 & 0.355 & \bf{53.5} & \bf{0.466} \\
 & & CD* & 60.5 & 41.2 & 14.5 & 4.7 & 0.412 & 0.316 & 63.5 & 0.345 \\
 & & PC & 50.0 & 5.6(18.3) & 21.9 & 4.2 & 0.239 & 0.085 & (80.3,  98.7) & (0.039,  0.191) \\
& {$(500, 100, 100)$} & CD & 104.2 & 81.7 & 17.1 & 5.5 & 0.816 & 0.217 & \bf{23.8} & \bf{0.666} \\
 & & CD* & 86.3 & 68.8 & 15.8 & 1.7 & 0.688 & 0.203 & 32.9 & 0.586 \\
 & & PC & 80.5 & 29.2(27.1) & 20.5 & 3.8 & 0.562 & 0.047 & (47.5, 74.6) & (0.193, 0.454) \\
\noalign{\smallskip}\hline\noalign{\smallskip}
Scale-free & {$(100, 100, 99)$} & CD & 99.2 & 74.8 & 17.9 & 6.5 & 0.756 & 0.245 & \bf{30.8} & 0.610 \\
 & & CD* & 103.3 & 76.3 & 18.1 & 8.8 & 0.771 & 0.260 & 31.5 & \bf{0.611} \\
 & & PC & 59.1 & 4.9(49.5) & 2.8 & 1.9 & 0.549 & 0.032 & (46.5, 96.0) & (0.032, 0.525) \\
& {$(500, 100, 99)$} & CD & 98.8 & 85.0 & 13.4 & 0.4 & 0.859 & 0.140 & \bf{14.5} & \bf{0.758} \\
 & & CD* & 105.2 & 85.5 & 13.6 & 6.2 & 0.863 & 0.186 & 19.8 & 0.726 \\
 & & PC & 74.0 & 1.2(71.0) & 0.0 & 1.8 & 0.729 & 0.023 & (28.6, 99.5) & (0.007, 0.717) \\
\noalign{\smallskip}\hline\noalign{\smallskip}
Small-world & {$(100, 100, 200)$} & CD & 145.2 & 69.1 & 49.1 & 26.9 & 0.346 & 0.523 & 157.9 & 0.250 \\
 & & CD* & 121.1 & 59.1 & 42.7 & 19.2 & 0.296 & 0.511 & 160.1 & 0.226 \\
 & & PC & 67.7 & 11.7(45.4) & 7.7 & 3.0 & 0.285 & 0.045 & ({\bf{146.0}}, 191.3) & (0.046, {\bf{0.271}}) \\
& {$(500, 100, 200)$} & CD & 168.8 & 98.8 & 53.0 & 17.0 & 0.494 & 0.412 & \bf{118.2} & \bf{0.367} \\
 & & CD* & 135.1 & 82.0 & 46.6 & 6.5 & 0.410 & 0.393 & 124.5 & 0.324 \\
 & & PC & 117.0 & 43.0(26.2) & 46.2 & 1.5 & 0.346 & 0.013 & (132.2, 158.4) & (0.158, 0.283) \\
\noalign{\smallskip}\hline\noalign{\smallskip}
Random DAG & {$(100, 100, 101.5)$} & CD & 91.8 & 56.6 & 24.7 & 10.4 & 0.556 & 0.384 & 55.4 & 0.414 \\
 & & CD* & 61.0 & 38.5 & 18.9 & 3.5 & 0.381 & 0.365 & 66.5 & 0.311 \\
 & & PC & 66.2 & 22.4(28.7) & 12.5 & 2.6 & 0.506 & 0.040 & ({\bf{53.0}}, 81.8) & (0.156, {\bf{0.441}}) \\
& {$(500, 100, 102.0)$} & CD & 103.4 & 76.2 & 23.4 & 3.8 & 0.746 & 0.263 & 29.8 & 0.591 \\
 & & CD* & 85.8 & 64.2 & 19.8 & 1.8 & 0.636 & 0.253 & 39.6 & 0.524 \\
 & & PC & 96.7 & 60.5(24.2) & 10.6 & 1.3 & 0.837 & 0.014 & ({\bf{18.6}}, 42.9) & (0.438, {\bf{0.755}}) \\
\noalign{\smallskip}\hline
\end{tabular}%
}
\end{center}
CD is the result of our CD algorithm with the smallest SHD along the solution path;
CD* is the result of our CD algorithm using our model selection criterion; 
The number in parentheses in column E for the PC algorithm reports the number of 
predicted undirected edges in the true skeleton 
\end{table*}

\begin{table*}[ht]
\begin{center}
\caption{Comparison between our CD algorithm and the EE-DAG sampler on simulated interventional data}
\label{tab:cmpEE}       
{\renewcommand{\arraystretch}{1}
\resizebox{\linewidth}{!}{%
\begin{tabular}{lclrrrrrrrr}
\hline\noalign{\smallskip}
Graph & $(n, p, s_0)$ & Method & P & E & R & FP & TPR & FDR &  SHD & JI \\
\noalign{\smallskip}\hline\noalign{\smallskip}
Bipartite &$(50, 50, 50)$ & CD & 30.8 & 19.4 & 6.5 & 4.8 & 0.389 & 0.362 & \bf{35.4} & 0.316 \\
 & & CD* & 29.6 & 17.8 & 6.5 & 5.3 & 0.355 & 0.398 & 37.6 & 0.286 \\
 & & EE & 53.9 & 27.4 & 4.8 & 21.6 & 0.548 & 0.486 & 44.2 & \bf{0.362} \\
\noalign{\smallskip}\hline\noalign{\smallskip}
Scale-free &$(50, 50, 49)$ & CD & 48.5 & 30.2 & 6.8 & 11.4 & 0.617 & 0.376 & \bf{30.1} & 0.452 \\
 & & CD* & 51.0 & 31.0 & 6.8 & 13.2 & 0.633 & 0.387 & 31.2 & \bf{0.454} \\
& & EE & 60.6 & 32.9 & 3.1 & 24.6 & 0.670 & 0.453 & 40.8 & 0.434 \\
\noalign{\smallskip}\hline\noalign{\smallskip}
Small-world &$(50, 50, 100)$ & CD & 70.2 & 35.0 & 25.9 & 9.2 & 0.350 & 0.498 & \bf{74.2} & \bf{0.260} \\
 & & CD* & 39.5 & 21.2 & 15.8 & 2.5 & 0.212 & 0.458 & 81.2 & 0.181 \\
& & EE & 62.4 & 30.6 & 16.1 & 15.7 & 0.306 & 0.507 & 85.0 & 0.234 \\
\noalign{\smallskip}\hline\noalign{\smallskip}
Random DAG &$(50, 50, 46.8)$ & CD & 34.1 & 20.4 & 8.3 & 5.5 & 0.441 & 0.400 & \bf{31.9} & 0.340 \\
 & & CD* & 26.1 & 16.4 & 6.8 & 2.9 & 0.361 & 0.363 & 33.2 & 0.296 \\
& & EE & 50.0 & 25.2 & 7.0 & 17.8 & 0.547 & 0.492 & 39.3 & \bf{0.358} \\
\noalign{\smallskip}\hline
\end{tabular}
}}
\end{center}
CD is the result of our CD algorithm with the smallest SHD along the solution path;
CD* is the result of our CD algorithm using our model selection criterion 
\end{table*}

Results reported in Table~\ref{tab:cmpPC} are the comparisons between our CD algorithm and the PC algorithm,  while in Table~\ref{tab:cmpEE} are the comparisons between our CD algorithm and the EE-DAG sampler. 
Results are averages over 20 data sets for each setting $(n, p, s_0)$. 
For our CD algorithm, we report two results for each setting, (i) result with the smallest SHD along the solution path; (ii) result using our model selection criterion (\ref{eq:dempiricalrule}) with $\alpha = 0.3$.
In Table~\ref{tab:cmpPC}, in order to make a clear comparison, we report lower and upper bounds of the SHD and the Jaccard index for CPDAGs estimated by the PC algorithm. Counting all undirected edges in a CPDAG that are in the true skeleton as expected edges, we will have a lower bound for the SHD and an upper bound for the  Jaccard index. Counting these undirected edges as reversed edges will give us an upper bound for the SHD and a lower bound for the Jaccard index.

It is obvious from Table~\ref{tab:cmpPC} that in majority of the cases our CD algorithm outperformed the PC algorithm. The SHD score is smaller than the lower bound of the PC algorithm, and the Jaccard index is higher than the upper bound of the PC algorithm. Only for random DAGs and small-world networks with $n=100$, the SHD of our CD algorithm is slightly higher than the lower bound of the SHD while the Jaccard index is slightly lower than the upper bound of the PC algorithm, 
showing that our algorithm was close to the best performance one could hope for the PC algorithm.
Note that when calculating the TPRs and FDRs in the table, the undirected edges are counted as expected ones, which clearly favors the PC algorithm.

We were only able to test the EE-DAG sampler on small graphs with $p = 50$ because its computing time was too long for graphs with $p = 100$. For the cases of $n =p=50$ (Table~\ref{tab:cmpEE}), we see that the smallest SHD our CD algorithm can achieve along a solution path is 20\% lower than the EE-DAG sampler for bipartite graphs, random DAGs and scale-free networks, and 10\% lower for small-world networks. The EE-DAG sampler predicted much more false positive edges (FP) but slightly fewer reversed edges (R) than our CD algorithm. For the cases of $n = 5p$, we report our comparisons in the Supplemental Material. The EE-DAG sampler had an outstanding performance, with a lower SHD and a higher Jaccard index. These observations are largely in agreement with our expectation. The EE-DAG sampler searches for
DAGs by sampling from a posterior distribution under the product multinomial model with a conjugate Dirichlet 
prior, which is close to $\ell_0$ regularization when $n$ is large.
Thus, it is expected to have good performance when $p$ is small and $n$ is large. However, the search is combinatorial in nature, which makes it impractical for even moderately large networks (such as the graphs with $p=100$ here).
On the contrary, our CD algorithm showed no problem in estimating graphs with hundreds of nodes and
can obtain comparable or better results in the cases of $n=p=50$.

We also did a comparison between the CD algorithm and the GIES algorithm, reported in the Supplemental Material. The results seem to suggest that the GIES algorithm often selects too many edges. When $n=p=100$, the number of edges the GIES algorithm predicted was around $3s_0$ in most of the cases. Consequently, it showed a much higher FDR as well as a larger SHD. 
Recall that the available package for the GIES algorithm can only take continuous data, which were generated by taking a transformation of the simulated discrete data. As a result, this comparison could be confounded by the use of different although related data sets, and thus is only intended to illustrate how the two algorithms would work. 

\begin{figure*}[t]
\centering
\includegraphics[width=0.75\linewidth]{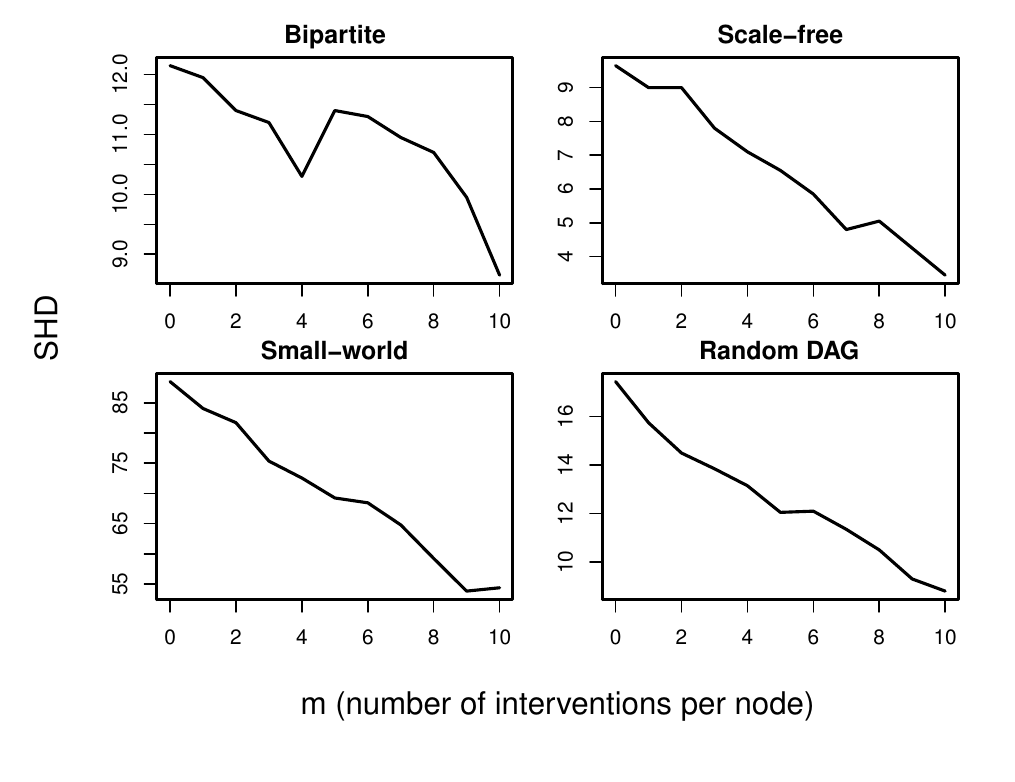}
\caption{The effect of interventions in terms of the SHD, where each node has $m$ interventional data points while the total sample size $n$ is fixed}
\label{fig:intervene}       
\end{figure*}

In order to show how interventional data improve the accuracy, we did more experiments. Figure~\ref{fig:intervene} shows how the SHD decreases when adding intervention to an observational data set, for all four types of graphs with $n=500$ and $p=50$. We started with a purely observational data set, and replaced $m$ observational data points by $m$ interventional data points for each node, for $m=1, 2, ..., 10$. The sample size was fixed as $n=500$. Therefore, we would finally have a data set with 10 interventions for each node. Figure~\ref{fig:intervene} shows the average of 20 experiments, with a very clear downward trend in all plots. The curve for the bipartite graph is not as smooth as the curves for the other types of networks. This is because the improvement for bipartite graphs is not as significant as the other networks. 

\subsection{Results for high-dimensional observational data}\label{sec:ExpwObs}

In this section, we apply our CD algorithm to high-dimensional observational data and compare its performance with the PC algorithm, the MMHC algorithm, and the HC algorithm. 

Metrics for estimation accuracy in this section are modified from those for interventional data. Since equivalent DAGs cannot be distinguished with observational data, we define reversed edges with regard to CPDAGs. A CPDAG is a partially directed graph that has all compelled (directed) edges in the equivalent class of a DAG. We calculate CPDAGs for both an estimated DAG and the true DAG. A reversed edge (R) refers to a predicted edge that satisfies the following two conditions: 
i) Its direction in the estimated DAG is wrong compared to the true DAG. ii) The direction of this edge is inconsistent between the CPDAGs of the true and estimated DAGs, including the case where the edge is directed in one CPDAG but undirected in the other.
Likewise, we define reversed edges in a CPDAG predicted by the PC algorithm as edges in the true skeleton that have an inconsistent direction with the true CPDAG.
The number of expected edges (E) is the number of estimated edges in the true skeleton excluding those reversed ones.  

For high-dimensional data, we generated graphs with $p = 200$ for each type of the networks. 
We chose $n = 50$ and the number of true edges $s_0$ ranged from 190 to 400 for these graphs.
Again, 20 data sets were generated for each combination of $(n, p)$. 

\begin{table*}[t]
\begin{center}
\caption{Comparison among our CD algorithm and other algorithms on simulated observational data}
\label{tab:highDim}       
{\renewcommand{\arraystretch}{1}
\resizebox{\linewidth}{!}{%
\begin{tabular}{lclrrrrccrr}
\hline\noalign{\smallskip}
 Network & $(n, p, s_0)$ & Method & P & E & R & FP & TPR & FDR & SHD & JI \\
\noalign{\smallskip}\hline\noalign{\smallskip}
Bipartite & {$(50, 200, 200.0)$} & CD & 108.7 & 69.6 & 20.6 & 18.6 & 0.348 & 0.357 & \bf{148.9} & \bf{0.290} \\
& & CD* & 90.2 & 59.6 & 17.8 & 12.8 & 0.298 & 0.333 & 153.2 & 0.258 \\
 & & PC & 75.7 & 26.9 & 34.2 & 14.6 & 0.134 & 0.643 & 187.7 & 0.108 \\
& & MMHC & 175.4 & 72.2 & 20.4 & 82.8 & 0.361 & 0.588 & 210.6 & 0.239 \\
 & & HC & 378.1 & 111.5 & 32.9 & 233.8 & 0.557 & 0.705 & 322.4 & 0.239 \\
\hline
Scale-free & {$(50, 200, 199.0)$} & CD & 139.6 & 83.8 & 15.8 & 39.9 & 0.421 & 0.402 & \bf{155.1} & 0.326 \\
 & & CD* & 201.8 & 106.8 & 21.6 & 73.4 & 0.537 & 0.470 & 165.6 & \bf{0.365} \\
 & & PC & 99.5 & 46.5 & 23.1 & 30.0 & 0.234 & 0.532 & 182.5 & 0.185 \\
& & MMHC & 176.8 & 93.2 & 16.1 & 67.5 & 0.468 & 0.472 & 173.2 & 0.330 \\
 & & HC & 377.8 & 121.0 & 28.1 & 228.8 & 0.608 & 0.680 & 306.9 & 0.266 \\
\hline 
Small-world & {$(50, 200, 400.0)$} & CD & 88.2 & 28.9 & 36.2 & 23.0 & 0.072 & 0.533 & \bf{394.1} & 0.058 \\
 & & CD* & 296.5 & 84.6 & 110.8 & 101.1 & 0.212 & 0.714 & 416.5 & 0.138 \\
 & & PC & 70.2 & 7.3 & 54.4 & 8.6 & 0.018 & 0.898 & 401.2 & 0.016 \\
& & MMHC & 249.2 & 84.3 & 85.7 & 79.3 & 0.211 & 0.661 & 395.0 & \bf{0.150} \\
  & & HC & 357.3 & 87.9 & 97.5 & 171.9 & 0.220 & 0.754 & 484.1 & 0.131 \\
\hline
Random DAG & {$(50, 200, 203.6)$} & CD & 113.0 & 68.8 & 28.1 & 16.1 & 0.339 & 0.386 & \bf{150.9} & 0.278 \\
& & CD* & 101.2 & 63.7 & 25.1 & 12.4 & 0.315 & 0.364 & 152.3 & 0.265 \\
 & & PC & 97.3 & 47.1 & 37.0 & 13.2 & 0.233 & 0.515 & 169.7 & 0.187 \\
& & MMHC & 179.8 & 86.4 & 31.1 & 62.4 & 0.427 & 0.520 & 179.6 & \bf{0.292} \\
 & & HC & 376.3 & 96.3 & 52.1 & 227.8 & 0.475 & 0.744 & 335.1 & 0.200 \\
\noalign{\smallskip}\hline
\end{tabular}
}}
\end{center}
CD* is the result of our CD algorithm using our model selection criterion 
\end{table*}

Table \ref{tab:highDim} summarizes the comparison results. 
One sees that our model selection criterion works quite well: 
The SHDs of CD and CD* in Table~\ref{tab:highDim} are very close. 
Our CD algorithm has the lowest SHD for all networks, and the Jaccard index is also quite high compared to other algorithms for most of the networks. 
When running the HC algorithm with the default setting, it tends to predict too many edges, which makes the comparison meaningless. For example, for a scale-free network with $p = 200$, the HC algorithm
predicted more than 2,000 edges.
Therefore, we set the maximum number of parents for each node to be 2 for the HC algorithm. 
However, it still predicted too many edges that it had a much higher TPR as well as a higher FDR and larger SHD than all the other algorithms. The PC algorithm predicted too few edges in scale-free networks and small-world networks, which led to a substantially lower TPRs and JIs, and the FDRs were quite high for bipartite networks. 
The MMHC algorithm, on the other hand, predicted a comparable number of edges as the true DAG in most cases. Its performance was in between of our CD algorithm and the PC algorithm. 
Our CD algorithm presents a clear advantage over all other algorithms in these high-dimensional cases.

\begin{figure*}[t]
\centering
 \includegraphics[width=0.8\linewidth]{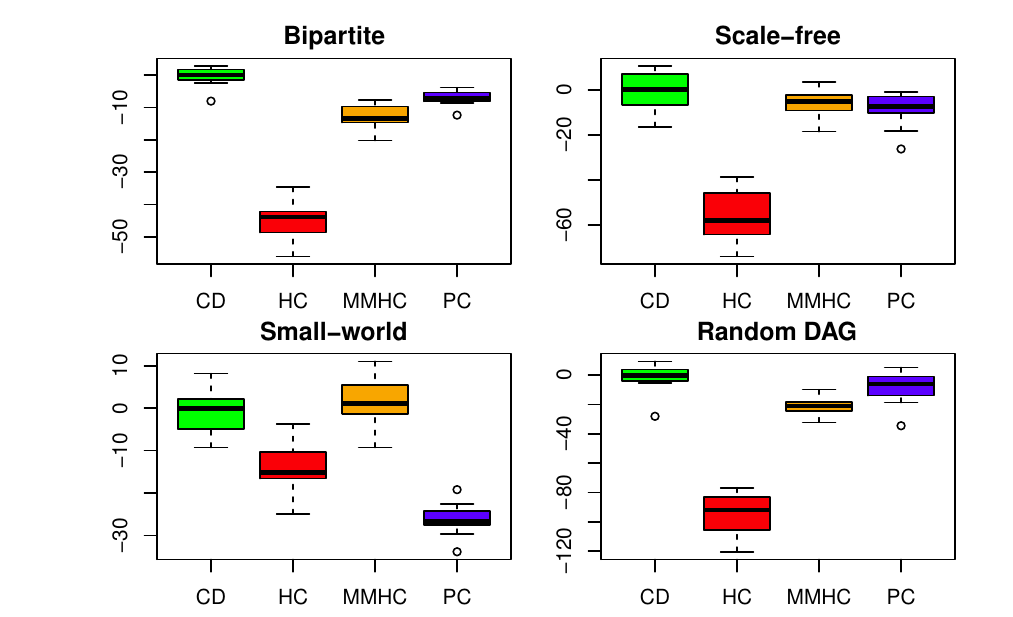}
\caption{Box-plot of test data log-likelihood for four algorithms with log-likelihood scaled by the sample size $n=50$ }
\label{fig:boxplot}       
\end{figure*} 

To further evaluate the quality of estimated networks, we computed test data log-likelihood to compare the predictive power. We generated 500 test data sets of the same size ($n=50$) for each DAG with $p=200$. We used each estimated graph to calculate the total log-likelihood of a test data set. 
Note that the output graph of the PC algorithm is a CPDAG, for which we cannot directly calculate the test data log-likelihood.  However, since the likelihood for any data set under every DAG in an equivalence class is the same, we converted a CPDAG output by the PC algorithm to an arbitrary DAG in the equivalence class and then calculated its test data log-likelihood.
Figure \ref{fig:boxplot} is the box-plot of test data log-likelihood for the four types of graphs in terms of the difference from the median of the test data log-likelihood of the CD algorithm. 
DAGs estimated by the CD algorithm were chosen by our model selection criterion.

It is seen from Figure \ref{fig:boxplot} that our CD algorithm has the highest test data log-likelihood for three out of the four types of networks. Only for small-world networks, our CD algorithm has a slightly lower log-likelihood than the MMHC algorithm. This shows that our method also has a very good predictive power in high-dimensional cases. We see that the HC algorithm has a much lower test data log-likelihood for most cases, which suggests overfitting given the observation that this algorithm often predicts too many edges. 
These results demonstrate the critical role of sparsity not only in structure estimation but also in predictive modeling.

\subsection{Timing comparison}\label{sec:timing}

We comment briefly on the comparison of running time among the algorithms.
The running time for generating a whole solution path for the interventional data
(results in Tables~\ref{tab:cmpPC} and \ref{tab:cmpEE}) was within 40 seconds
for all graphs.
The PC algorithm was about two times faster than our CD algorithm for bipartite graphs and random DAGs, but it did not scale well for scale-free networks and small-world networks. For these two types of networks, running time of the PC algorithm was much longer for $n=500$ and $p=100$.
For the high-dimensional data in Table~\ref{tab:highDim}, it took between 2 and 20 seconds for our method to compute
the entire solution path. The speed of the PC algorithm was quite comparable to our CD algorithm on these data sets. The MMHC algorithm was faster which took at most 5 seconds for all data sets. 
The fastest algorithm was the HC algorithm, however, its accuracy in learning Bayesian networks was too bad so we will not go into details for this algorithm.
Our method gives a principled way
to incorporate interventional data and is often more accurate than the other competitors.
These merits in performance justify its utility. In addition, there is room for a more efficient implementation
of our algorithm which may improve its speed substantially.

\section{Applications to Real Networks}\label{sec:realNetworks}

In this section, we apply our CD algorithm to real networks. In Section \ref{sec:compK2}, we examine how the proposed multi-logit model compares to the product multinomial model by comparing our CD algorithm to the K2 algorithm \citep{Cooper92}. We will then apply our method to a real data set  in Section \ref{sec:dreal}.

\subsection{Comparison with the K2 Algorithm} \label{sec:compK2}

The K2 algorithm  is a well-known method for learning discrete Bayesian networks based on a product multinomial model. However, it requires an input ordering of the nodes. A wrong ordering can severely damage the quality of the estimated graph. Therefore, we provide the K2 algorithm with an ordering that is compatible with the true DAG to obtain the best estimation.  In order to conduct a fair comparison, we also input the same ordering to our CD algorithm by only running the inner loop of Algorithm \ref{algo:linesearch}, which is equivalent to a sequence of $p-1$ penalized multi-logit regression problems. With a known ordering, a main difference between the two algorithms is the underlying model, the multi-logit model for our CD algorithm and the multinomial model for the K2 algorithm. We used a Matlab package \texttt{K2} \citep{bielza2011multi} to run the algorithm. The K2 algorithm also requires an upper bound for the maximum number of parents for each node. In our experiments, we set the upper bound to be 4. We chose 8 real networks provided by the \texttt{bnlearn} package, where $p$ ranges from 8 to 441. Observational data were simulated from these networks, and for each DAG, 20 data sets were generated independently according to a product multinomial model. This comparison will demonstrate how well our proposed multi-logit model approximates the multinomial model.

\begin{table*}[ht]
\begin{center}
\caption{Comparison between our CD algorithm and the K2 Algorithm}
\label{tab:cmpK2}       
{\renewcommand{\arraystretch}{1}
\resizebox{\linewidth}{!}{%
\begin{tabular}{r l | r c c r r | r c c r r }
\hline\noalign{\smallskip}
& & \multicolumn{5}{|c|}{CD Algorithm} & \multicolumn{5}{c} {K2 Algorithm}\\
Network & ${(n, p, s_0)}$ & P & TPR & FDR & SHD & JI & P & TPR & FDR & SHD & JI  \\
\noalign{\smallskip}\hline\noalign{\smallskip}
 asia & $(250, 8, 8)$ & 10.2 & 0.719 & 0.430 & 6.7 & 0.469 & 10.1 & 0.838 & 0.331 & \bf{4.7} & \bf{0.597} \\
 sachs & $(250, 11, 17)$ & 14.5 & 0.732 & 0.133 & \bf{6.6} & \bf{0.659} & 14.7 & 0.538 & 0.374 & 13.3 & 0.408 \\
 child & $(250, 20, 25)$ & 31.1 & 0.656 & 0.469 & \bf{23.3} & \bf{0.416} & 30.1 & 0.602 & 0.500 & 25.0 & 0.376 \\
 insurance & $(250, 27, 52)$ & 50.4 & 0.473 & 0.511 & \bf{53.2} & \bf{0.316} & 51.1 & 0.414 & 0.578 & 60.0 & 0.265 \\
 alarm & (250, 37, 46) & 60.8 & 0.664 & 0.497 & \bf{45.6} & \bf{0.401} & 60.8 & 0.618 & 0.531 & 49.9 & 0.364 \\
 hailfinder & $(250, 56, 66)$ & 80.2 & 0.525 & 0.558 & 76.9 & 0.313 & 79.1 & 0.546 & 0.542 & \bf{73.0} & \bf{0.331} \\
 hepar2 & $(250, 70, 123)$ & 137.6 & 0.269 & 0.756 & \bf{194.5} & \bf{0.146} & 139.9 & 0.236 & 0.792 & 204.8 & 0.124 \\
 pigs & $(250, 441, 592)$ & 773.65 & 0.863 & 0.334 & \bf{343.5} & \bf{0.600} & 788.8 & 0.704 & 0.472 & 547.4 & 0.432 \\
\noalign{\smallskip}\hline
\end{tabular}
}}
\end{center}
\end{table*}

Summary of the comparison for the 8 networks is provided in Table \ref{tab:cmpK2}. Since a correct ordering is given, there will not be any reversed edges, and thus, in this table, we only report P, TPR, FDR, SHD, and JI. Here we matched the number of predicted edges of our CD algorithm with the K2 algorithm. It can be seen that for most graphs the SHD for our CD algorithm is lower than that of the K2 algorithm, while the JI is higher, except the networks asia and hailfinder. Since the data sets were simulated by product multinomial models, this result confirms that our proposed multi-logit model serves as a good approximation to the full multinomial model. This comparison also suggests that the group norm regularization in our method may be more efficient than using an upper bound on the parent size as in the K2 algorithm. 

\subsection{Application to flow cytometry data}\label{sec:dreal}

We consider in this section applying the CD algorithm to a real data set
that has been extensively studied. The data set was generated from a flow cytometry experiment conducted
by \citet{Sachs05}, who studied a well-known signaling network
in human primary CD4+ T-cells of the immune system.
This chosen network was perturbed by various stimulatory and
inhibitory interventions. Each interventional condition was applied to
an individual component of the network. Simultaneous measurements were
taken on $p=11$ proteins and phospholipids of this network from individual cells under each
condition. Since three interventions were targeted at proteins that were
not measured, samples collected under these conditions were
observational. Among the 11 measured
components, five proteins and phospholipids were perturbed. 
The data set contains measurements for $n=5,400$ cells. 
Each variable has three levels (high, medium and low), and
consequently, the size of a component group of $\bmbeta$ is 6 for this data set.

\begin{figure*}[ht]
\centering
\begin{minipage}[b]{0.45\linewidth}
\centering
\includegraphics[width=\linewidth, trim = 20mm 25mm -10mm -5mm]{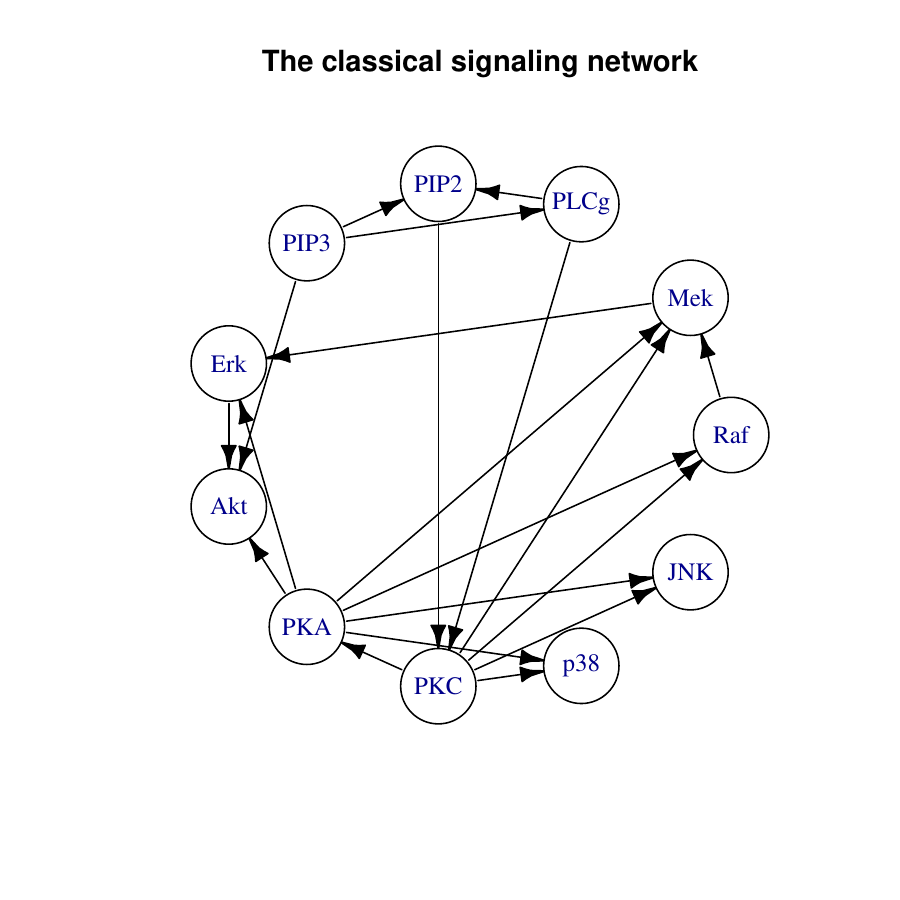} \\
(A)
\end{minipage}
\begin{minipage}[b]{0.45\linewidth}
\centering
\includegraphics[width=\linewidth, trim = 20mm 25mm -10mm -5mm]{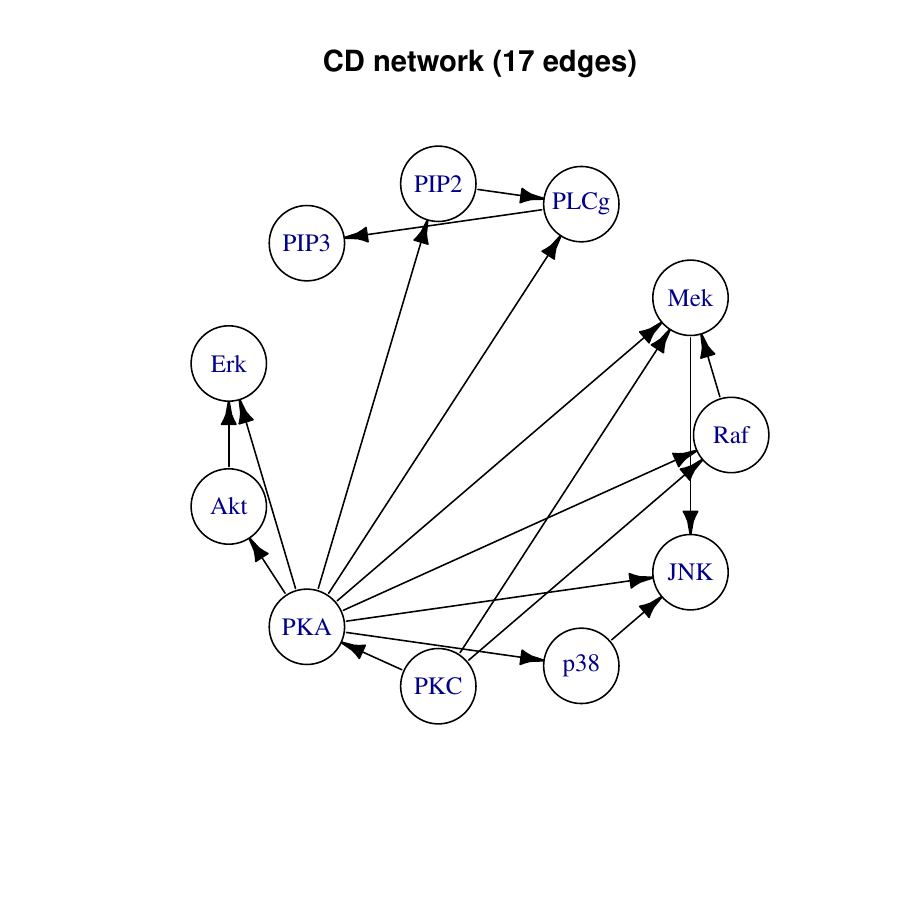} \\ 
(B)
\end{minipage}
\caption{(A) The consensus signaling network in human immune system cells, (B) DAG estimated by the CD algorithm}
\label{fig:cyto}       
\end{figure*}

Figure~\ref{fig:cyto}A is a plot for the known causal
interactions among the 11 components of this signaling network. These
causal relationships are well-established, and no consensus has been
reached on interactions beyond those present in the network. 
This network structure is often used as the benchmark to assess the
accuracy of an estimated network. Therefore, we call it the consensus
model. Our estimated network by the CD algorithm with the smallest SHD along the solution path is shown in
Figure~\ref{fig:cyto}B. The DAG is qualitatively
close to the consensus model. More detailed performance measures are reported
in Table~\ref{tab:expReal}, including both results for the DAG with the smallest SHD (CD algorithm) and the one selected by our model selection criterion (CD algorithm*).  As a comparison, we include the DAGs 
estimated by three competing methods, the order-graph sampler \citep{Ellis08}, 
the EE-DAG sampler, and the PC algorithm. 
Our  CD algorithm showed a very competitive performance, predicting more or comparable number of expected edges
and fewer reversed edges than the other three methods. 
Our method also had the least number of false positive edges among all the methods. 
All these led to the lowest SHD and highest Jaccard index for our CD algorithm.
Note that for the PC algorithm, we counted all 3 undirected edges in the true skeleton as expected edges
in the calculation of the SHD and JI. Yet 
our CD algorithm still outperformed it.
Markov chain Monte Carlo (MCMC) methods for DAG estimation 
often have good performance when the number of nodes $p$ is small, but they do not scale well.
Thus, it is comforting to see that our method, which can handle larger networks, 
outperforms MCMC methods on this relatively small network. 

\begin{table}[ht]
\begin{center}
\caption{Comparison on the flow cytometry data set}
\label{tab:expReal}       
{\renewcommand{\arraystretch}{1}
\begin{tabular}{c|rrrrrrr}
\hline\noalign{\smallskip}
Method & P & E & R & M & FP & SHD & JI  \\
\noalign{\smallskip}\hline\noalign{\smallskip}
CD algorithm & 17 & 10 & 3 & 7 & 4 & \bf{14} & \bf{0.370} \\
CD algorithm* & 14 & 8 & 3 & 9 & 3 & 15 & 0.308 \\
PC algorithm & 17 & 5(3) & 4 & 8 & 5 & 17 & 0.276 \\
Order-graph sampler & 20 & 8 & 4 & 8 & 8 & 20 & 0.250 \\
EE-DAG sampler & 26 & 9 & 6 & 5 & 11 & 22 & 0.243 \\
\noalign{\smallskip}\hline
\end{tabular}
}
\end{center}
The order-graph sampler result comes from the mean graph (Figure
11 in \citealp{Ellis08})
\end{table}


\section{Discussions}\label{sec:dsummary}

We have developed a maximum penalized likelihood method for estimating sparse
discrete Bayesian networks under a multi-logit model. In order to avoid 
penalizing separately individual dummy variables for a
factor, a group norm penalty is utilized to
encourage sparsity at the factor level. A blockwise
coordinate descent algorithm is developed where each coordinate
descent step is solved by iteratively applying a quadratic
approximation. The acyclicity constraint imposed on the structure of
Bayesian networks can be enforced in a natural way by the coordinate descent
algorithm. Our method has been evaluated on simulated graphs 
and real-world networks, with both interventional and observational data. 
We have demonstrated that
 the CD algorithm outperforms many existing methods,
particularly when $n\leq p$. We have also performed an analysis of a flow
cytometry data set generated from a signaling
network in human immune system cells. The DAG estimated by
the CD algorithm is close to the consensus model. Since the true network is not
available, the estimated edges provide candidate
causal interactions that could be tested in future experiments. 
 
Computation for estimating discrete Bayesian
networks is demanding due to the size of the parameter space and the
nonlinear nature of the multi-logit model. There is room for improving the efficiency 
of the CD algorithm. For example, one may incorporate the idea of stochastic gradient 
descent in the quadratic approximation step, which will reduce significantly the computation.
Moreover, since our search space is nonconvex,
introducing such components of stochastic optimization may also increase the chance of
finding a global minimizer of the penalized loss function. 
Other future directions include studying the consistency of our penalized estimator when
the number of nodes $p=p_n$ grows with the sample size $n$ and investigating 
the use of group concave penalties.


\section*{Appendix: Asymptotic Theory}\label{sec:dAP}

In this Appendix, we establish asymptotic theory for the DAG estimator 
$\hat{\boldsymbol{\beta}}_{\lambda}$ (\ref{eq:argminBeta}) assuming that $p$ is fixed and $n\to\infty$.
By rearranging and relabeling individual components, we rewrite $\boldsymbol{\beta}$ as
$\boldsymbol{\phi}=(\boldsymbol{\phi}_{(1)}, \boldsymbol{\phi}_{(2)})$, where
$\boldsymbol{\phi}_{(1)} = \text{vec}(
\boldsymbol{\beta}_{1\cdot 1},\ldots, 
\boldsymbol{\beta}_{1\cdot p},\ldots \ldots,
\boldsymbol{\beta}_{p\cdot 1},\ldots,
\boldsymbol{\beta}_{p\cdot p})$ is the parameter vector of
interest and $\boldsymbol{\phi}_{(2)} = \text{vec}(\boldsymbol{\beta}_{1\cdot 0},
\ldots,\boldsymbol{\beta}_{p\cdot 0})$ denotes the vector of
intercepts. Hereafter, we
denote by $\phi_j$ the $j^{\text{th}}$ group of $\boldsymbol{\phi}$,
such that $\phi_1 = \boldsymbol{\beta}_{1\cdot 1}$,  $\phi_2 = \boldsymbol{\beta}_{1\cdot 2}$, 
$\ldots, \phi_{p^2}=\bmbeta_{p\cdot p}$, and so on.
We say $\boldsymbol{\phi}$ is acyclic if the graph $\mathcal{G}_{\boldsymbol{\phi}}$
induced by $\boldsymbol{\phi}$ (or the corresponding $\boldsymbol{\beta}$) is
acyclic. 

Define $\bmphi_{[k]}$ ($k\in\{1,\ldots,p\}$)
to be the parameter vector obtained from $\boldsymbol{\phi}$ by setting
$\boldsymbol{\beta}_{k\cdot i} = \mathbf{0}$ for $i=1, \ldots, p$. 
In other words,  the DAG $\mathcal{G}_{\bmphi_{[k]}}$ is obtained by
deleting all edges pointing to the
$k^{\text{th}}$ node in $\mathcal{G}_{\boldsymbol{\phi}}$; see (\ref{eq:sld}).
We assume the data set $\mathcal{X}$ consists of $(p+1)$ 
blocks, denoted by $\mathcal{X}^j$ of size $n_j \times p$, $j=1,\ldots,p+1$. The node $X_j$ is
experimentally fixed in $\mathcal{X}^j$ for the first $p$ blocks, 
while the last block contains purely observational data.
Let $\mathcal{I}_j$ be the set of row indices of $\mathcal{X}^j$.
As demonstrated by \eqref{eq:BNint}, we can model
interventional data in the
$k^{\text{th}}$ block of the data matrix $\mathcal{X}^k$ as \textit{i.i.d.} observations
from a joint distribution factorized according to
$\mathcal{G}_{\bmphi_{[k]}}$. Denote the corresponding
probability mass function by $p(\mathbf{x}|\bmphi_{[k]})$, where
$\mathbf{x}=(x_1,\ldots,x_p)$ and $x_j \in \{ 1, \ldots, r_j \}$ for
$j=1,\ldots,p$. 
To simplify our notation, denote the parameter for the $(p+1)$th block by $\bmphi_{[p+1]} = \bmphi$.
Then the log-likelihood of $\mathcal{X}$ is
\begin{align}\label{eq:loglik}
L(\boldsymbol{\phi}) = \sum_{k=1}^{p+1} L_k(\bmphi_{[k]}) 
= \sum_{k=1}^{p+1} \log p(\mathcal{X}^k \mid \bmphi_{[k]}),
\end{align}
where $\log p(\mathcal{X}^k | \bmphi_{[k]})=\sum_{h \in \mathcal{I}_k} \log(p(\mathcal{X}_{h\cdot}|
\bmphi_{[k]}))$ and $\mathcal{X}_{h\cdot}$ $=(\mathcal{X}_{h1},\ldots,\mathcal{X}_{hp})$. 
The penalized log-likelihood function with a tuning parameter $\lambda_n>0$ is
\begin{eqnarray}\label{eq:rl}
R(\boldsymbol{\phi})
&=& L(\boldsymbol{\phi})-\lambda_n\sum_{j=1}^{p^2}\norm{\phi_j}_2 \nonumber \\
&=& \sum_{k=1}^{p+1}L_k(\bmphi_{[k]})-\lambda_n\sum_{j=1}^{p^2}\norm{\phi_j}_2,
\end{eqnarray}
where the component group $\phi_j$ $(j=1,\ldots,p^2)$
represents the influence of one variable on another. 
Let $\Upomega = \{\boldsymbol{\phi}: \mathcal{G}_{\boldsymbol{\phi}} \text{ is a DAG}\}$ be the parameter space.
A penalized estimator $\hat{\boldsymbol{\phi}}$ is obtained by
maximizing $R(\boldsymbol{\phi})$ in
$\Upomega$.

Though interventional data help distinguish equivalent DAGs, the
following notion of natural parameters is needed to completely establish
identifiability of DAGs
for the case where each variable has interventional data. 
We say that $i$ is an ancestor of $j$
in a DAG $\mathcal{G}$ if there exists at least one path
from $i$ to $j$. Denote the set of ancestors of $j$
by $\text{an}(j)$.

\begin{definition}[Natural parameters]
We say that $\boldsymbol{\phi} \in \Upomega$ is natural if
$i \in  \text{an}(j) \text{ in } \mathcal{G}_{\boldsymbol{\phi}}$ implies that
$j$ is not independent of $i$ under the joint distribution given by $\bmphi_{[i]}$
for all $i,j=1,\ldots,p$.
\end{definition}

For a causal DAG, a natural
parameter implies that the effects along
multiple causal paths connecting the same pair of nodes do not cancel. 
This is a reasonable assumption for many real-world problems,
and is much weaker than the faithfulness assumption. 
Under the faithfulness assumption, all conditional independence restrictions
can be read off from $d$-separations in the DAG. If nodes $i$ and $j$ are independent
in $\bmphi_{[i]}$, then by faithfulness the nodes $i$ and $j$ must be separated by empty
set and thus $i \notin  \text{an}(j)$ in $\mathcal{G}_{\boldsymbol{\phi}_{[i]}}$.
This implies that $i \notin  \text{an}(j)$ in $\mathcal{G}_{\boldsymbol{\phi}}$ as well,
by the construction of $\mathcal{G}_{\boldsymbol{\phi}_{[i]}}$. Indeed, we see that the faithfulness
assumption implies the natural parameter assumption.

To establish asymptotic properties of our
penalized likelihood estimator, we make the following assumptions:
\begin{itemize}
\item[(A1)] The true parameter $\boldsymbol{\phi}^*$ is natural
and an interior point of $\Upomega$.
\item[(A2)] The parameter $\boldsymbol{\theta}_j$ of the conditional distribution
$[X_j | \Uppi_j^{\mathcal{G}}; \boldsymbol{\theta}_j]$ is identifiable for each $j=1,\ldots,p$. 
The log-likelihood function $\ell_j(\boldsymbol{\theta}_j) = \log p({x}_j|\Uppi_j^{\mathcal{G}}; \boldsymbol{\theta}_j)$ is strictly concave and continuously three times differentiable for any interior point.
\end{itemize}
Recall that the $k^{\text{th}}$ block of our data, $\mathcal{X}^k$, can be regarded as
an \textit{i.i.d.} sample of size $n_k$ from the distribution $p(\mathbf{x}|\bmphi_{[k]}^*)$
for all $k$, while we define $\bmphi_{[p+1]}^*=\bmphi^*$ for the last block of observational data.
 
\begin{theorem}\label{didentifiable}
Assume (A1) and (A2). If
$p(\mathbf{x}|\bmphi_{[k]})=p(\mathbf{x}|\bmphi_{[k]}^*)$ for all possible $\mathbf{x}$ and
all $k=1,\ldots,p$, then $\boldsymbol{\phi}=\boldsymbol{\phi}^*$.
Furthermore, if $n_k\gg \sqrt{n}$ for all $k=1,\ldots,p$,
then for any $\boldsymbol{\phi} \neq \boldsymbol{\phi}^*$,
\begin{equation}\label{eq:didentifiability2}
P(L(\boldsymbol{\phi}^*)>L(\boldsymbol{\phi})) \rightarrow 1 \;\;\;\text{ as } n \to \infty.
\end{equation}
\end{theorem}

\begin{theorem}\label{dOLSconsistency}
Assume (A1) and (A2). If $\lambda_n/\sqrt{n}\rightarrow 0$ and $n_k\gg \sqrt{n}$ for all $k=1,\ldots,p$,
then there exists a global maximizer $\hat{\boldsymbol{\phi}}$ of 
$R(\boldsymbol{\phi})$ such that 
$\norm{\hat{\boldsymbol{\phi}}-\boldsymbol{\phi}^*}_2=O_p(n^{-1/2})$.
\end{theorem}

Proofs of the two theorems are relegated to the Supplemental Material.
Theorem~\ref{didentifiable} confirms that the causal DAG model is identifiable with interventional data
assuming a natural parameter. Theorem~\ref{dOLSconsistency} implies that there is a $\sqrt{n}$-consistent
global maximizer of $R(\bmphi)$ with the group norm penalty. Note that Assumption (A2) does not specify
a particular choice of model for the conditional distribution $[X_j | \Uppi_j^{\mathcal{G}}]$
and thus these theoretical results apply to a large class of DAG models for discrete data.
In particular, the multi-logit regression model (\ref{eq:mlogit}) satisfies (A2).
 
\begin{remark}
The assumption on the sample size of interventional data, $n_k\gg \sqrt{n}$, imposes
a lower bound on how fast the fraction $\alpha_k=n_k/n\gg n^{-1/2}$ can approach zero for $k=1,\ldots,p$.
Although this allows the observational data to dominate when $\alpha_k\to 0$, the fractions of interventional data
must be larger than the typical order $O_p(n^{-1/2})$ of statistical errors so that \eqref{eq:didentifiability2}
can hold to establish identifiability of the true causal DAG parameter $\boldsymbol{\phi}^*$.
This guarantees that the global maximizer $\hat{\boldsymbol{\phi}}$ 
will locate in a neighborhood of $\boldsymbol{\phi}^*$
with high probability. Once in this neighborhood, the convergence rate
of $\hat{\boldsymbol{\phi}}$ then depends on the size $n$ of all data, both interventional and observational. 
Therefore, increasing the size of observation data will lead to more accurate estimate $\hat{\boldsymbol{\phi}}$
as long as we keep $\alpha_k\gg n^{-1/2}$ for $k=1,\ldots,p$.
\end{remark} 

\begin{remark}
It is interesting to generalize the above asymptotic results to the case where
$p=p_n$ grows with the sample size $n$, say, by developing nonasymptotic bounds on 
the $\ell_2$ estimation error
$\|\hat{\bmphi}-\bmphi^*\|_2$. However, in order to estimate the causal network
consistently, sufficient interventional data are needed for each node, i.e., $n_k$ must approach infinity,
and thus $p/ n \to 0$ as $n\to \infty$. This limits us to the low-dimensional setting with $p<n$. Suppose
we have a large network with $p\gg n$. 
One may first apply some regularization method on observational data 
to screen out independent nodes and to partition the network
into small subgraphs that are disconnected to one another. Then for each small subgraph, we can afford
to generate enough interventional data for every node and apply the method in this paper to infer the causal structure. 
Our asymptotic theory provides useful guidance for the analysis in the second step.
\end{remark}

For purely observational data, the theory becomes more complicated due to the existence
of equivalent DAGs and parameterizations. 
It is left as future work to establish the consistency of a global maximizer for high-dimensional observational data.

\bibliographystyle{asa}
\bibliography{references}

\begin{thebibliography}{51}
\newcommand{\enquote}[1]{``#1''}
\expandafter\ifx\csname natexlab\endcsname\relax\def\natexlab#1{#1}\fi

\bibitem[{Aragam and Zhou(2015)}]{Aragam14}
Aragam, B. and Zhou, Q. (2015), \enquote{Concave penalized estimation of sparse
  {B}ayesian networks,} \textit{Journal of Machine Learning Research}, 16,
  2273--2328.

\bibitem[{Barab\'{a}si and Albert(1999)}]{Barabasi99}
Barab\'{a}si, A.-L. and Albert, R. (1999), \enquote{Emergence of Scaling in
  Random Networks,} \textit{Science}, 286, 509--512.

\bibitem[{Bielza et~al.(2011)Bielza, Li, and Larranaga}]{bielza2011multi}
Bielza, C., Li, G., and Larranaga, P. (2011), \enquote{Multi-dimensional
  classification with {B}ayesian networks,} \textit{International Journal of
  Approximate Reasoning}, 52, 705--727.

\bibitem[{Bouckaert(1993)}]{Bouckaert93prob}
Bouckaert, R.~R. (1993), \enquote{Probabilistic Network Construction Using the
  Minimum Description Length Principle,} in \textit{Symbolic and Quantitative
  Approaches to Reasoning and Uncertainty: European Conference ECSQARU '93},
  Springer, vol. 747 of \textit{Lecture Notes in Computer Science}, pp. 41--48.

\bibitem[{Bouckaert(1994)}]{Bouckaert94prob}
--- (1994), \enquote{Probabilistic Network Construction Using the Minimum
  Description Length Principle,} Tech. Rep. RUU-CS-94-27, Department of
  Computer Science, Utrecht University.

\bibitem[{Buntine(1991)}]{Buntine91}
Buntine, W. (1991), \enquote{Theory Refinement on {B}ayesian Networks,} in
  \textit{Proceedings of the Seventh Annual Conference on Uncertainty in
  Artificial Intelligence}, Morgan Kaufmann Publishers Inc., pp. 52--60.

\bibitem[{Chickering and Heckerman(1997)}]{Chickering97}
Chickering, D.~M. and Heckerman, D. (1997), \enquote{Efficient Approximations
  for the Marginal Likelihood of {B}ayesian Networks with Hidden Variables,}
  \textit{Machine Learning}, 29, 181--212.

\bibitem[{Cooper and Herskovits(1992)}]{Cooper92}
Cooper, G.~F. and Herskovits, E. (1992), \enquote{A {B}ayesian Method for the
  Induction of Probabilistic Networks from Data,} \textit{Machine Learning}, 9,
  309--347.

\bibitem[{Cooper and Yoo(1999)}]{cooper1999causal}
Cooper, G.~F. and Yoo, C. (1999), \enquote{Causal discovery from a mixture of
  experimental and observational data,} in \textit{Proceedings of the Fifteenth
  conference on Uncertainty in artificial intelligence}, Morgan Kaufmann
  Publishers Inc., pp. 116--125.

\bibitem[{Cs\'{a}rdi and Nepusz(2006)}]{Csardi06}
Cs\'{a}rdi, G. and Nepusz, T. (2006), \enquote{The igraph Software Package for
  Complex Network Research,} \textit{InterJournal: Complex Systems}, 1695.

\bibitem[{Ellis and Wong(2008)}]{Ellis08}
Ellis, B. and Wong, W.~H. (2008), \enquote{Learning Causal {B}ayesian Network
  Structures from Experimental Data,} \textit{Journal of the American
  Statistical Association}, 103, 778--789.

\bibitem[{Erdos and R{\'e}nyi(1960)}]{renyi1959random}
Erdos, P. and R{\'e}nyi, A. (1960), \enquote{On the evolution of random
  graphs,} \textit{Publ. Math. Inst. Hung. Acad. Sci}, 5, 17--60.

\bibitem[{Friedman et~al.(2007)Friedman, Hastie, H{\"o}fling, and
  Tibshirani}]{Friedman07}
Friedman, J., Hastie, T., H{\"o}fling, H., and Tibshirani, R. (2007),
  \enquote{Pathwise Coordinate Optimization,} \textit{The Annals of Applied
  Statistics}, 1, 302--332.

\bibitem[{Friedman et~al.(2010)Friedman, Hastie, and Tibshirani}]{Friedman10}
Friedman, J., Hastie, T., and Tibshirani, R. (2010), \enquote{Regularization
  Paths for Generalized Linear Models via Coordinate Descent,} \textit{Journal
  of Statistical Software}, 33, 1--22.

\bibitem[{Fu and Zhou(2013)}]{Fu13}
Fu, F. and Zhou, Q. (2013), \enquote{Learning Sparse Causal {G}aussian Networks
  With Experimental Intervention: Regularization and Coordinate Descent,}
  \textit{Journal of the American Statistical Association}, 108, 288--300.

\bibitem[{Fu(1998)}]{Fu98}
Fu, W. (1998), \enquote{Penalized Regressions: The Bridge versus the Lasso,}
  \textit{Journal of Computational and Graphical Statistics}, 7, 397--416.

\bibitem[{G{\'a}mez et~al.(2011)G{\'a}mez, Mateo, and
  Puerta}]{gamez2011learning}
G{\'a}mez, J.~A., Mateo, J.~L., and Puerta, J.~M. (2011), \enquote{Learning
  {B}ayesian networks by hill climbing: efficient methods based on progressive
  restriction of the neighborhood,} \textit{Data Mining and Knowledge
  Discovery}, 22, 106--148.

\bibitem[{Han et~al.(2016)Han, Chen, Cheon, and Zhong}]{han2016estimation}
Han, S.~W., Chen, G., Cheon, M.-S., and Zhong, H. (2016), \enquote{Estimation
  of Directed Acyclic Graphs Through Two-Stage Adaptive Lasso for Gene Network
  Inference,} \textit{Journal of the American Statistical Association}, 111,
  1004--1019.

\bibitem[{Hauser and B\"uhlmann(2012)}]{hauser2012learning}
Hauser, A. and B\"uhlmann, P. (2012), \enquote{Characterization and greedy
  learning of interventional {M}arkov equivalence classes of directed acyclic
  graphs,} \textit{Journal of Machine Learning Research}, 13, 2409--2464.

\bibitem[{Hauser and B{\"u}hlmann(2015)}]{hauser2015jointly}
Hauser, A. and B{\"u}hlmann, P. (2015), \enquote{Jointly interventional and
  observational data: estimation of interventional Markov equivalence classes
  of directed acyclic graphs,} \textit{Journal of the Royal Statistical
  Society: Series B (Statistical Methodology)}, 77, 291--318.

\bibitem[{Heckerman et~al.(1995)Heckerman, Geiger, and
  Chickering}]{Heckerman95}
Heckerman, D., Geiger, D., and Chickering, D.~M. (1995), \enquote{Learning
  {B}ayesian Networks: The Combination of Knowledge and Statistical Data,}
  \textit{Machine Learning}, 20, 197--243.

\bibitem[{Herskovits and Cooper(1990)}]{Herskovits90}
Herskovits, E. and Cooper, G. (1990), \enquote{Kutat\'{o}: An Entropy-Driven
  System for Construction of Probabilistic Expert Systems from Databases,} in
  \textit{Proceedings of the Sixth Annual Conference on Uncertainty in
  Artificial Intelligence}, Elsevier Science Inc., pp. 117--128.

\bibitem[{Kalisch and B{\"u}hlmann(2007)}]{kalisch2007estimating}
Kalisch, M. and B{\"u}hlmann, P. (2007), \enquote{Estimating high-dimensional
  directed acyclic graphs with the PC-algorithm,} \textit{The Journal of
  Machine Learning Research}, 8, 613--636.

\bibitem[{Kalisch et~al.(2012)Kalisch, M\"achler, Colombo, Maathuis, and
  B\"uhlmann}]{kalisch2012causal}
Kalisch, M., M\"achler, M., Colombo, D., Maathuis, M.~H., and B\"uhlmann, P.
  (2012), \enquote{Causal Inference Using Graphical Models with the {R} Package
  {pcalg},} \textit{Journal of Statistical Software}, 47, 1--26.

\bibitem[{Kou et~al.(2006)Kou, Zhou, and Wong}]{kou2006discussion}
Kou, S., Zhou, Q., and Wong, W.~H. (2006), \enquote{Equi-energy sampler with
  applications in statistical inference and statistical mechanics (with
  discussion),} \textit{The Annals of Statistics}, 34, 1581--1652.

\bibitem[{Lam and Bacchus(1994)}]{Lam94}
Lam, W. and Bacchus, F. (1994), \enquote{Learning {B}ayesian Belief Networks:
  An Approach Based on the {MDL} Principle,} \textit{Computational
  Intelligence}, 10, 269--293.

\bibitem[{Lee et~al.(2016)Lee, Simchowitz, Jordan, and Recht}]{lee2016gradient}
Lee, J.~D., Simchowitz, M., Jordan, M.~I., and Recht, B. (2016),
  \enquote{Gradient descent only converges to minimizers,} vol.~49, pp. 1--12.

\bibitem[{Meganck et~al.(2006)Meganck, Leray, and
  Manderick}]{meganck2006learning}
Meganck, S., Leray, P., and Manderick, B. (2006), \enquote{Learning causal
  {B}ayesian networks from observations and experiments: A decision theoretic
  approach,} in \textit{International Conference on Modeling Decisions for
  Artificial Intelligence}, Springer, pp. 58--69.

\bibitem[{Meier et~al.(2008)Meier, van~de Geer, and B{\"u}hlmann}]{Meier08}
Meier, L., van~de Geer, S., and B{\"u}hlmann, P. (2008), \enquote{The Group
  Lasso for Logistic Regression,} \textit{Journal of the Royal Statistical
  Society. Series B (Statistical Methodology).}, 70, 53--71.

\bibitem[{Pearl(2003)}]{Pearl00}
Pearl, J. (2003), \enquote{Causality: Models, Reasoning, and Inference,}
  \textit{Econometric Theory}, 19, 675--685.

\bibitem[{Pe{\'e}r et~al.(2001)Pe{\'e}r, Regev, Elidan, and
  Friedman}]{pe2001inferring}
Pe{\'e}r, D., Regev, A., Elidan, G., and Friedman, N. (2001),
  \enquote{Inferring subnetworks from perturbed expression profiles,}
  \textit{Bioinformatics}, 17, S215--S224.

\bibitem[{Pournara and Wernisch(2004)}]{pournara2004reconstruction}
Pournara, I. and Wernisch, L. (2004), \enquote{Reconstruction of gene networks
  using {B}ayesian learning and manipulation experiments,}
  \textit{Bioinformatics}, 20, 2934--2942.

\bibitem[{Sachs et~al.(2005)Sachs, Perez, Pe\'er, Lauffenburger, and
  Nolan}]{Sachs05}
Sachs, K., Perez, O., Pe\'er, D., Lauffenburger, D.~A., and Nolan, G.~P.
  (2005), \enquote{Causal Protein-Signaling Networks Derived from
  Multiparameter Single-Cell Data,} \textit{Science}, 308, 523--529.

\bibitem[{Schmidt and Murphy(2006)}]{schmidt2006lassoordersearch}
Schmidt, M. and Murphy, K. (2006), \enquote{LassoOrderSearch: Learning Directed
  Graphical Model Structure using $\ell_1$-Penalized Regression and Order
  Search,} \textit{Learning}, 8, 2.

\bibitem[{Schmidt et~al.(2007)Schmidt, Niculescu-Mizil, Murphy,
  et~al.}]{schmidt2007learning}
Schmidt, M., Niculescu-Mizil, A., Murphy, K., et~al. (2007), \enquote{Learning
  graphical model structure using $\ell_1$-regularization paths,} in
  \textit{AAAI}, vol.~7, pp. 1278--1283.

\bibitem[{Scutari(2010)}]{scutari2010learning}
Scutari, M. (2010), \enquote{Learning {B}ayesian Networks with the {bnlearn}
  {R} Package,} \textit{Journal of Statistical Software}, 35, 1--22.

\bibitem[{Scutari(2016)}]{scutari2016empirical}
--- (2016), \enquote{An Empirical-Bayes Score for Discrete Bayesian Networks,}
  in \textit{Conference on Probabilistic Graphical Models}, pp. 438--448.

\bibitem[{Scutari(2017)}]{Scutari2017learnining}
--- (2017), \enquote{Bayesian Network Constraint-Based Structure Learning
  Algorithms: Parallel and Optimized Implementations in the {bnlearn} {R}
  Package,} \textit{Journal of Statistical Software}, 77, 1--20.

\bibitem[{Shojaie et~al.(2014)Shojaie, Jauhiainen, Kallitsis, and
  Michailidis}]{shojaie2014inferring}
Shojaie, A., Jauhiainen, A., Kallitsis, M., and Michailidis, G. (2014),
  \enquote{Inferring regulatory networks by combining perturbation screens and
  steady state gene expression profiles,} \textit{PloS one}, 9, e82393.

\bibitem[{Shojaie and Michailidis(2010)}]{Shojaie10}
Shojaie, A. and Michailidis, G. (2010), \enquote{Penalized Likelihood Methods
  for Estimation of Sparse High-Dimensional Directed Acyclic Graphs,}
  \textit{Biometrika}, 97, 519--538.

\bibitem[{Spirtes et~al.(1993)Spirtes, Glymour, and Scheines}]{Spirtes93}
Spirtes, P., Glymour, C., and Scheines, R. (1993), \textit{Causation,
  Prediction, and Search}, Springer-Verlag.

\bibitem[{Suzuki(1993)}]{Suzuki93}
Suzuki, J. (1993), \enquote{A Construction of {B}ayesian Networks from
  Databases Based on an {MDL} Principle,} in \textit{Proceedings of the Ninth
  Annual Conference on Uncertainty in Artificial Intelligence}, pp. 266--273.

\bibitem[{Tsamardinos et~al.(2006)Tsamardinos, Brown, and
  Aliferis}]{tsamardinos2006max}
Tsamardinos, I., Brown, L.~E., and Aliferis, C.~F. (2006), \enquote{The max-min
  hill-climbing {B}ayesian network structure learning algorithm,}
  \textit{Machine learning}, 65, 31--78.

\bibitem[{Tseng and Yun(2009)}]{Tseng09}
Tseng, P. and Yun, S. (2009), \enquote{A Coordinate Gradient Descent Method for
  Nonsmooth Separable Minimization,} \textit{Mathematical Programming}, 117,
  387--423.

\bibitem[{van~de Geer and B{\"u}hlmann(2013)}]{geer2013}
van~de Geer, S. and B{\"u}hlmann, P. (2013), \enquote{$\ell_0$-penalized
  maximum likelihood for sparse directed acyclic graphs,} \textit{The Annals of
  Statistics}, 41, 536--567.

\bibitem[{Venables and Ripley(2002)}]{nnet2002}
Venables, W.~N. and Ripley, B.~D. (2002), \textit{Modern Applied Statistics
  with S}, New York: Springer, 4th ed., iSBN 0-387-95457-0.

\bibitem[{Watts and Strogatz(1998)}]{Watts98}
Watts, D.~J. and Strogatz, S.~H. (1998), \enquote{Collective Dynamics of
  `Small-World' Networks,} \textit{Nature}, 393, 440--442.

\bibitem[{Wu and Lange(2008)}]{Wu08}
Wu, T. and Lange, K. (2008), \enquote{Coordinate Descent Algorithms for Lasso
  Penalized Regression,} \textit{The Annals of Applied Statistics}, 2,
  224--244.

\bibitem[{Yuan and Lin(2006)}]{Yuan06}
Yuan, M. and Lin, Y. (2006), \enquote{Model Selection and Estimation in
  Regression with Grouped Variables,} \textit{Journal of the Royal Statistical
  Society. Series B (Statistical Methodology)}, 68, 49--67.

\bibitem[{Zhou(2011)}]{Zhou11}
Zhou, Q. (2011), \enquote{Multi-Domain Sampling with Applications to Structural
  Inference of {B}ayesian Networks,} \textit{Journal of the American
  Statistical Association}, 106, 1317--1330.

\bibitem[{Zhu and Hastie(2004)}]{Zhu04}
Zhu, J. and Hastie, T. (2004), \enquote{Classification of Gene Microarrays by
  Penalized Logistic Regression,} \textit{Biostatistics}, 5, 427--443.

\end{thebibliography}

\end{document}